\documentclass[10pt,twocolumn,twoside]{IEEEtran}

\usepackage{graphicx}      % include this line if your document contains figures
\usepackage{amsmath}
\usepackage{amssymb}
\usepackage{graphicx}
\usepackage{color}
\usepackage{bm}
\usepackage{dsfont}
\usepackage{cite}
\usepackage{epstopdf}
\usepackage{stfloats}
\newtheorem{theorem}{Theorem}[section]
\newtheorem{lemma}[theorem]{Lemma}

\newtheorem{remark}{Remark}[section]
\newenvironment{proof}{\noindent{\bf Proof}\,\,}{$\hfill\blacksquare$\\}

\begin{document}

\title{Identification of FIR Systems with Binary Input and Output Observations} 

\author{Alex S. Leong, Erik Weyer, and Girish N. Nair
% \thanks{Copyright (c) 2012 IEEE. Personal use of this material is permitted. However, permission to use this material for any other purposes must be obtained from the IEEE by sending a request to pubs-permissions@ieee.org.}
 \thanks{A. Leong is with the Department of Electrical Engineering (EIM-E), Paderborn University, 33098 Paderborn, Germany. E-mail: {\tt alex.leong@upb.de.} E. Weyer and G. Nair are with the
  Department of Electrical and Electronic Engineering,
  University of Melbourne, Parkville, Vic. 3010, Australia. E-mail: {\tt \{ewey,gnair\}@unimelb.edu.au.}}
\thanks{This work was supported by the Australian Research Council under grants DP120101122 and DE120102012.}
  }

\maketitle

\begin{abstract}
This paper considers the identification of FIR systems, where information about the inputs and outputs of the system undergoes quantization into binary values before transmission to the estimator. In the case where the thresholds of the input and output quantizers can be adapted, but the quantizers have no  computation and storage capabilities, we propose identification schemes which are strongly consistent for Gaussian distributed inputs and noises. This is based on exploiting the correlations between the quantized input and output observations to derive nonlinear equations that the true system parameters must satisfy, and then estimating the parameters by solving these equations using stochastic approximation techniques. If, in addition, the input and output quantizers have computational and storage capabilities, strongly consistent identification schemes are proposed which can handle arbitrary input and noise distributions. In this case, some conditional expectation terms are computed at the quantizers, which can then be estimated based on binary data transmitted by the quantizers, subsequently allowing the parameters to be identified by solving a set of linear equations. The algorithms and their properties are illustrated in simulation examples.
%An extension of our methods to estimate the parameters of ARX systems is also considered. 
\end{abstract}

\section{Introduction}
It  is nowadays  common to transmit data using digital communication techniques rather than analog communications, due to advantages such as better noise tolerance and the possibility of doing error control coding on the data  \cite{Proakis}. In digital communications,  analog valued data is required to be quantized into a digital form (e.g. bit strings of 0s and 1s) before transmission.  For applications such as large-scale production plants and environmental monitoring, the sensors must transmit their measurements over a communication network to a distant monitoring station. Unlike consumer internet and telephony, these networks must often satisfy severe limitations on transmission power and bandwidth, for reasons of cost and energy efficiency \cite{Heinzelman_HICSS,Akyildiz}. This thus limits the resolution in bits of the
transmitted measurements and degrades the quality of the models and relationships constructed from the received data.
In this paper we  consider  identification of FIR systems where information about the inputs and outputs of the system are quantized to a single bit (i.e. binary data) at each discrete time instant before transmission to the estimator.

System identification using quantized observations has been previously studied. The case where only the system outputs are quantized has been considered in e.g. \cite{Wang_book,AgueroGoodwinYuz,WeyerKoCampi,GodoyGoodwin,CasiniGarulliVicino_CDC,You_adaptive} by using  multi-level quantizers, and  \cite{Wang_book,WangZhangYin,ColinetJuillard,CasiniGarulliVicino,CsajiWeyer_CDC,CsajiWeyer_sysid,GoudjilPouliquen,PouliquenGoudjil}
by using 1-bit quantizers. Different aspects such as asymptotic properties of estimators, design of input signals for identification, and design of quantizers and threshold selection have been investigated, with various  assumptions made on the type of system and level of knowledge of the noise distributions.

 In this paper we do not assume that the input signal can be designed, but we  assume that we have quantized measurements of it.
Such a situation is encountered in many areas, e.g. process industries, ecology, environmental sciences, and economics  \cite[p.409]{Ljung_book} where one cannot or should not interfere with the system (or certain parts of the system). In such cases the input signal is measured  rather than specifically designed.  This is also different from the setting in blind system identification as studied in the signal processing and communication literature, where the system parameters are identified (up to a multiplicative constant) based on only statistical information about the input signal in addition to measurements of the output \cite{AbedMeraimQiuHua,YuZhangXie}. 

When both the inputs and outputs are quantized with multi-level quantizers, approaches using  instrumental variables methods  were proposed in \cite{SuzukiSugie,Ikenoue}, but the analysis relies on the validity of high rate quantization assumptions \cite{GrayNeuhoff} and no proof of consistency was provided. The problem of finding an optimal fixed order FIR approximation from quantized input and output data was studied in \cite{CeronePigaRegruto}. 
In the case where both inputs and outputs are quantized to 1-bit, \cite{Krishnamurthy_quantized} studied the identification of a  dynamic shock error model by counting patterns of zeros and ones \cite{Kedem}, which can give consistent estimates for known noise distributions, but requires knowledge of the power ratio between the input and output signals. The identification of FIR systems where output observations are quantized, and the input signal is constrained to take on a finite number of possible values, was studied in \cite{GuoWang_FIR}. The identification of first order gain systems with binary input and output observations was investigated in \cite{YouWeyerNair} in the case where the noise and inputs were assumed to be Gaussian, and for  symmetrically distributed (about its mean) inputs and noises in \cite{LianLuo}. Identification schemes based on empirical measures and the EM algorithm were presented in \cite{YouWeyerNair}, and schemes based on stochastic approximation in \cite{LianLuo}.  However, for higher order FIR systems with input and output observations quantized to a single bit, no consistent identification schemes currently exist. Other related work include identification of Wiener \cite{Chen_Wiener_discontinuous,MuChen_Wiener}, Hammerstein \cite{ZhaoChen_Hammerstein,MuChenWangYinZheng}, and nonlinear ARX \cite{ZhaoZhengBai,ZhaoChenTempoDabbene} systems, but the inputs are  assumed to be perfectly known in these works, and only \cite{ZhaoChenTempoDabbene} explicitly considers quantized outputs.

In this paper we extend the setup considered in \cite{YouWeyerNair} and \cite{LianLuo} to FIR systems. Their proposed methods however do not  generalize in a straightforward manner to higher order systems, thus alternative identification schemes are devised. 
The main contributions of the paper are:
\begin{itemize}
\item We consider identification of FIR systems where \emph{both} the input and output observations are quantized to 1-bit.  The input signal is not designed, and we only have information about the realization of the signal from the quantized measurements.
\item In the case where the thresholds of the input and output quantizers can be dynamically adjusted, but the quantizers have no  computation and storage capabilities,  we propose identification schemes which are strongly consistent for i.i.d. Gaussian distributed inputs and noises. 
\item  If, in addition, the input and output quantizers have computational and storage capabilities, we devise strongly consistent identification schemes for arbitrary i.i.d. input and noise distributions. 
\end{itemize}

The paper is organized as follows. Section \ref{dumb_quantizer_sec} considers identification of FIR systems for quantizers without computational capabilities and Gaussian distributed  inputs and noise. We present identification schemes  when either the parameters of the Gaussian distributed input are known (Section \ref{dumb_quantizer_scheme_sec}) or unknown (Section \ref{dumb_quantizer_unknown_input_sec}), together with proofs of strong consistency of the schemes. The idea is based on exploiting the correlations between the quantized input and output observations to derive nonlinear equations that the parameters must satisfy.  The parameters are then estimated by solving these nonlinear equations using stochastic approximation techniques. Section \ref{smart_quantizer_sec} considers the case of quantizers with computational and storage capabilities, and arbitrary input and noise distributions. Identification schemes are presented when either the input distribution is known (Section \ref{smart_quantizer_scheme_sec}) or unknown (Section \ref{smart_quantizer_unknown_input_sec}), together with proofs of strong consistency. The idea  is now to compute certain conditional expectation terms at the quantizers. These conditional expectations  can  be estimated  based on binary data transmitted by the quantizers, which then allows the parameters to be identified by solving a set of linear equations. A preliminary version of the results in this paper (without convergence proofs) can be found in \cite{LeongWeyerNair}. 
%The identification of parameters in ARX systems is considered in Section \ref{ARX_sec}. 

\section{Quantizers Without Computational Capabilities}
\label{dumb_quantizer_sec}
 
\subsection{Data Generating System and Model}
\label{dumb_quantizer_model_sec}
The system to be identified is an $N$-th order FIR system
\begin{equation*}
%\label{FIR_model}
y_t = b_1 u_{t-1} + b_2 u_{t-2} + \dots + b_N u_{t-N} + w_t
\end{equation*}
where $\{u_t\}$ are the inputs, $\{y_t\}$ the outputs, $\{w_t\}$ the noise, and $b_1,\dots,b_N$ are the parameters to be identified. 
There are quantizers at the inputs $\{u_t\}$ and outputs $\{y_t\}$, which transmit 1-bit (binary) quantized information to the  estimator, see Fig. \ref{system_model}. 
\begin{figure}[t!]
\centering 
\includegraphics[scale=0.4]{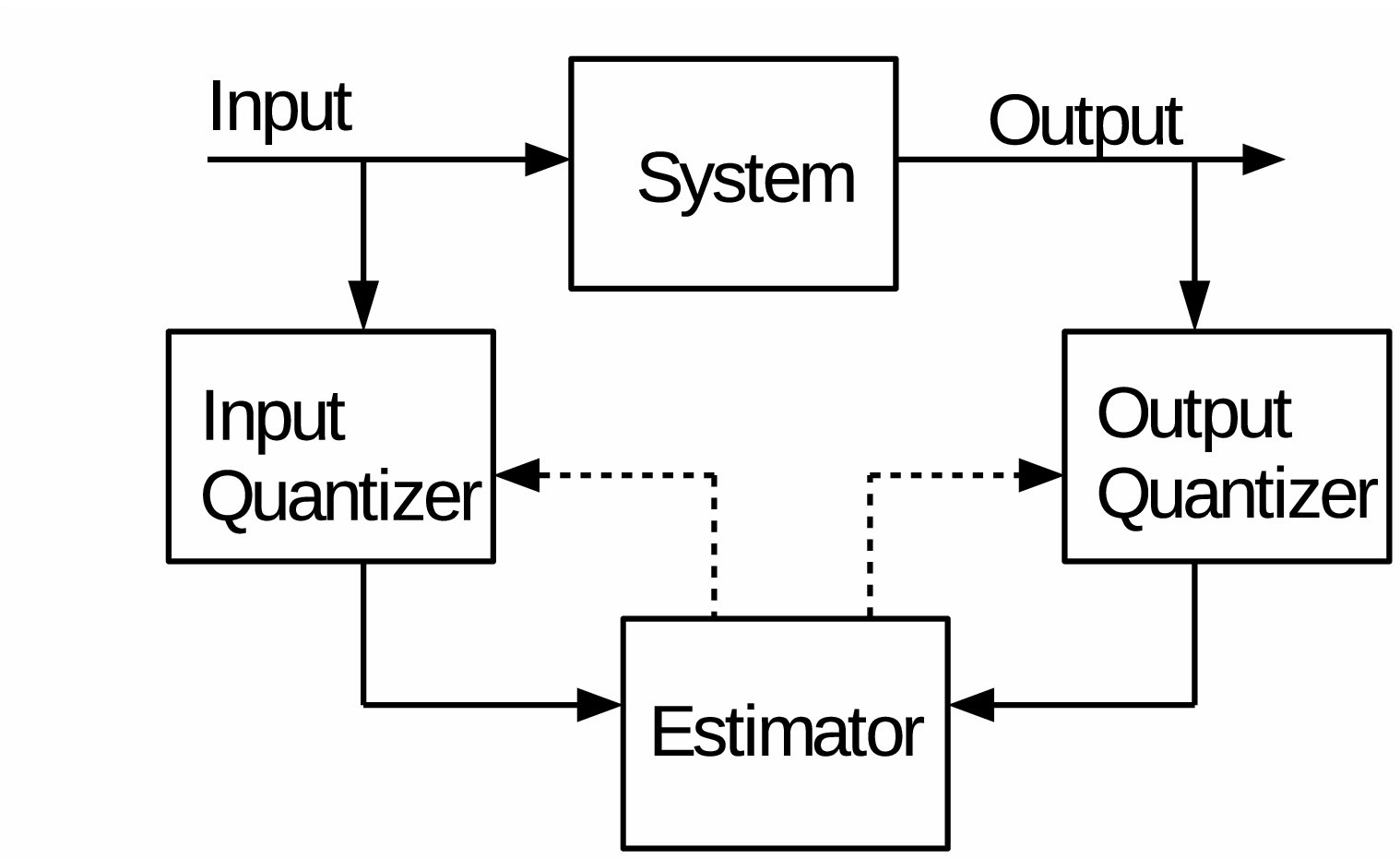} 
\caption{System Model}
\label{system_model}
\end{figure} 
The estimator can also transmit information back to the input and output quantizers, e.g. it can tell the quantizers  to  adjust their thresholds. However, in this section the quantizers will be assumed to have no additional computation or storage capabilities. At time $t$, the estimator will receive the measurements $\mathbf{z}_t = (\mathds{1}(u_{t} >c_u), \mathds{1} (y_t > c_y))$, where $\mathds{1}(.)$ is the indicator function, and $c_u$ and $c_y$ are the input and output quantizer thresholds respectively. 
In this paper we are primarily interested in FIR systems of order $N \geq 2$, as the case of $N=1$ with binary input and output observations has been previously studied in \cite{YouWeyerNair} and \cite{LianLuo}. 

We make the following assumptions:
\\ \emph{Assumption 1}: The input sequence $\{u_t\}$ is i.i.d. Gaussian with mean $\mu$ and variance $\sigma_u^2$.
\\ \emph{Assumption 2}: The noise sequence $\{w_t\}$ is i.i.d. Gaussian and independent of $\{u_t\}$, with zero mean and variance $\sigma_w^2$. 
\\ \emph{Assumption 3}: The model order $N$ is known.

We will assume that $\mu$ and  $\sigma_u^2$  are known to the estimator in Section \ref{dumb_quantizer_scheme_sec}, but unknown in Section \ref{dumb_quantizer_unknown_input_sec}. It will turn out that knowledge of the noise variance $\sigma_w^2$ is not needed in the identification schemes. In Section \ref{smart_quantizer_sec}, the Gaussian assumptions stated in Assumptions 1 and 2 will be removed.

\subsection{Identification Scheme for Known Input Distribution}
\label{dumb_quantizer_scheme_sec}
In this subsection we will also make the following assumption:
\\ \emph{Assumption 4}: The input parameters $\mu$ and  $\sigma_u^2$  are known to the estimator. 

We will first describe the intuition behind the identification scheme, before presenting it formally in Algorithm 1. We will then give a proof of the strong consistency of the identification scheme. 

The basic idea is to consider the correlations between the quantized input and output observations. Specifically, we look at the product $\mathds{1}(u_{t-n} >c_u) \mathds{1} (y_t > c_y)$ for  $n=1,\dots,N$. Taking the empirical mean, we have by the ergodic theorem (see e.g. p. 393 of \cite{GrimmettStirzaker}) that as $T \rightarrow \infty$,
\begin{equation}
\label{empirical_corr}
\begin{split}
&\frac{1}{T} \sum_{t=1}^T \mathds{1}(u_{t-n} > c_u) \mathds{1} (y_t > c_y)  \stackrel{\textrm{a.s.}}{\rightarrow} \mathbb{P} (u_{t-n} > c_u, y_t > c_y) \\
 & = \mathbb{P} (u_{t-n} > c_u, b_1 u_{t-1}  + \dots + b_N u_{t-N} + w_t > c_y) \\
 & = \mathbb{E} \left[ \mathbb{P} (u_{t-n} > c_u, b_1 u_{t-1}  + \dots + b_N u_{t-N} + w_t > c_y | u_{t-n}) \right] \\
 & = \!\int_{c_u}^\infty \! \Bigg[1 \!-\! \Phi \Bigg( \frac{c_y \!-\! b_n u_{t-n} \!-\! \sum_{m=1,m\neq n}^N b_m \mu}{\sqrt{\sum_{m=1,m\neq n}^N b_m^2 \sigma_u^2 + \sigma_w^2}}\Bigg) \Bigg] p(u_{t-n}) du_{t-n}  
\end{split}
\end{equation}
where $\Phi(x) = \int_{-\infty}^x \frac{1}{\sqrt{2\pi}} e^{-t^2/2} dt$ is the cumulative distribution function (cdf) of a $\mathcal{N}(0,1)$ random variable, and
\begin{equation}
\label{pu_defn}
p(u) \triangleq \frac{1}{\sqrt{2 \pi \sigma_u^2}} \exp \left(- \frac{(u-\mu)^2}{2 \sigma_u^2} \right)
\end{equation}
is the probability density function (pdf) of a $\mathcal{N}(\mu,\sigma_u^2)$ random variable.  
The last line of (\ref{empirical_corr}) holds since given $u_{t-n}$, $ b_1 u_{t-1}  +\dots + b_N u_{t-N} + w_t$ is Gaussian with mean $b_n u_{t-n} + (b_1 + \dots + b_{n-1}  + b_{n+1} + \dots + b_N) \mu$ and variance $(b_1^2 + \dots + b_{n-1}^2 + b_{n+1}^2 + \dots + b_N^2) \sigma_u^2 + \sigma_w^2$. 
Let $y$ be a random variable with the same stationary distribution as $y_t$.
Substituting the expressions $\mathbb{E}[y] = (b_1 + \dots + b_N) \mu$ and $\textrm{Var}[y] = (b_1^2 + \dots+ b_N^2) \sigma_u^2 + \sigma_w^2$ into (\ref{empirical_corr}) gives 
\begin{equation}
\label{empirical_prob_expr}
\begin{split}
&\frac{1}{T} \sum_{t=1}^T \mathds{1}(u_{t-n} > c_u) \mathds{1} (y_t > c_y) \\ & \stackrel{\textrm{a.s.}}{\rightarrow} \int_{c_u}^\infty \Bigg[1 - \Phi \Bigg( \frac{c_y - b_n u - \mathbb{E}[y] + b_n \mu}{\sqrt{\textrm{Var}[y] - b_n^2 \sigma_u^2}}\Bigg) \Bigg] p(u) du   
\end{split}
\end{equation}
%where the last line holds if we use the quantizer thresholds $c_u=\mu$ and $c_y = \mathbb{E}[y]$. 
The idea is now to estimate $\mathbb{E}[y]$ and $\textrm{Var}[y]$, and to substitute these estimates in the equations above and solve with respect to $b_n$.

The identification scheme is divided into odd and even time slots.\footnote{Devoting half the resources to estimating the mean and half  to estimating the variance is an intuitively reasonable choice. Whether there is a different proportion that gives ``optimal'' performance will however require further investigation.} During the odd time slots $t=2j-1, j=1,2,\dots$, we estimate $\mathbb{E}[y]$,  by using the stochastic approximation (\!\!\!\cite{KushnerYin, Chen_SA}) procedure
\begin{equation}
\label{cy_recursion}
c_{y,j+1} = c_{y,j} + \alpha_j \left( \mathds{1} (y_{2j-1} > c_{y,j}) -0.5 \right),
\end{equation}
  where $\{\alpha_j\}$ is a  sequence satisfying $\alpha_j > 0$, $\sum_{j=0}^\infty \alpha_j = \infty$,  and $\sum_{j=0}^\infty \alpha_j^2 < \infty$. 
  The procedure tries to find a $c_y$ such that 
$\mathbb{P}(y_t > c_y) = 0.5$, so that the estimate of the mean is $\widehat{\textrm{E}y} = c_y$, since the probability that a random variable is larger than its mean is $0.5$ for any symmetric distribution with a continuous pdf such as the Gaussian. To see that (\ref{cy_recursion}) is a stochastic approximation procedure, write
\begin{align*} 
& \mathds{1} (y_{2j-1} > c_{y,j}) -0.5 \\ &=\mathbb{P} (y_t > c_{y,j}) -0.5 + \mathds{1} (y_{2j-1} > c_{y,j}) - \mathbb{P} (y_t > c_{y,j}).
\end{align*}
Thus $\mathds{1} (y_{2j-1} > c_{y,j}) -0.5$ can be regarded as a ``noisy''  observation (with noise term $\mathds{1} (y_{2j-1} > c_{y,j}) - \mathbb{P} (y_t > c_{y,j})$) of the function $\mathbb{P} (y_t > c_{y,j}) -0.5$, whose root we are trying to find. 

During the even time slots $t=2j, j=1,2,\dots$, we estimate $\textrm{Var}[y] = (b_1^2 + \dots + b_N^2) \sigma_u^2 + \sigma_w^2$, using the stochastic approximation procedure
 $$\tilde{c}_{y,j+1} = \tilde{c}_{y,j} + \alpha_j \left(  \mathds{1}(y_{2j} > \tilde{c}_{y,j}) - 0.1587 \right).$$ 
This procedure tries to find a $\tilde{c}_y$ such that 
$\mathbb{P}(y_t > \tilde{c}_y) = 0.1587$. 
 Since $y_t$ is Gaussian, it follows that $\tilde{c}_y$ will be one standard deviation larger than the mean, since the probability that a Gaussian random variable is more than one standard deviation away from the mean is $1-0.6827 = 0.1587 \times 2 $. Hence an estimate of the variance is $\widehat{\textrm{V}y} =  (\tilde{c}_y - \widehat{\textrm{E}y})^2 = (\tilde{c}_y - c_y)^2.  $

%Estimation of the parameters $b_1,\dots,b_N$ also occurs during the odd time slots. 
Replacing $\mathbb{E}[Y]$ with $c_y$ and $\textrm{Var}[Y]$ with $\widehat{\textrm{V}y}$ on the right hand side of (\ref{empirical_prob_expr}),  and choosing the threshold\footnote{The choice $c_u=\mu$ gives roughly equal proportions of $0$'s and $1$'s for the random variable $\mathds{1}(u_t > c_u)$, though any other reasonably chosen value for $c_u$ will work. A similar comment applies to the choice of $\tilde{c}_y$ to be one standard deviation larger than the mean. These choices in the algorithm could possibly be tweaked and optimized over, but in this paper we will use intuitively natural values to illustrate the basic principles.} $c_u=\mu$, gives the equations
\begin{equation}
\label{bn_eqn}
\begin{split}
& \frac{1}{T} \sum_{t=1}^T \mathds{1}(u_{t-n} > \mu) \mathds{1} (y_t > c_y) \\ & = \int_{\mu}^\infty \Bigg[1 - \Phi \Bigg( \frac{- b_n(u - \mu)}{\sqrt{\widehat{\textrm{V}y} - b_n^2 \sigma_u^2}}\Bigg) \Bigg] p(u) du, \quad n=1,\dots,N,
\end{split}
\end{equation}
which can be solved with respect to $b_n$, thereby obtaining estimates. 

Note that each of the $N$  equations in (\ref{bn_eqn}) is an equation of one variable, and all the equations are of the same form. 
A question arises as to whether each of the equations in (\ref{bn_eqn}) has a unique solution for $b_n$.
Given $\mu$ and $\sigma_u^2$, define 
\begin{equation}
\label{F_fn}
\begin{split}
&F(b,\widehat{\textrm{V}y}) \\ & \triangleq \left\{ \begin{array}{ll} \int_{\mu}^\infty \Big[1 - \Phi \Big( \frac{- b(u - \mu)}{\sqrt{\widehat{\textrm{V}y} - b^2 \sigma_u^2}}\Big) \Big] p(u) du, & - \sqrt{\frac{\widehat{\textrm{V}y}}{\sigma_u^2}} < b < \sqrt{\frac{\widehat{\textrm{V}y}}{\sigma_u^2}} \\ \int_{\mu}^\infty p(u) du = \frac{1}{2}, & b \geq \sqrt{\frac{\widehat{\textrm{V}y}}{\sigma_u^2}} \\ 0, &  b \leq - \sqrt{\frac{\widehat{\textrm{V}y}}{\sigma_u^2}}. \end{array} \right.
\end{split}
\end{equation}
\begin{lemma}
\label{F_increasing_lemma}
For fixed $\widehat{\textrm{V}y}$, the function $F(b,\widehat{\textrm{V}y})$ defined by (\ref{F_fn}) is strictly monotonically increasing in $b$ for  $ b \in \left(- \sqrt{\widehat{\textrm{V}y}/\sigma_u^2},  \sqrt{\widehat{\textrm{V}y}/\sigma_u^2}\right)$. 
\end{lemma}
\begin{proof}
See Appendix \ref{F_increasing_lemma_proof}.
\end{proof}

\begin{remark}
Since
$$ \lim_{b \rightarrow \sqrt{\widehat{\textrm{V}y}/\sigma_u^2}}\int_{\mu}^\infty \Big[1 - \Phi  \Big( \frac{- b(u - \mu)}{\sqrt{\widehat{\textrm{V}y} - b^2 \sigma_u^2}}\Big) \Big] p(u) du  = \int_{\mu}^\infty  p(u) du$$
and $$ \lim_{b \rightarrow - \sqrt{\widehat{\textrm{V}y}/\sigma_u^2}}\int_{\mu}^\infty \Big[1 - \Phi  \Big( \frac{- b(u - \mu)}{\sqrt{\widehat{\textrm{V}y} - b^2 \sigma_u^2}}\Big) \Big] p(u) du  = 0,$$
 $F(b,\widehat{\textrm{V}y})$ is monotonically increasing in $b$ for fixed $\widehat{\textrm{V}y}$, and strictly monotonic on the interval $\left(- \sqrt{\widehat{\textrm{V}y}/\sigma_u^2},  \sqrt{\widehat{\textrm{V}y}/\sigma_u^2}\right)$ as shown in Lemma \ref{F_increasing_lemma}.
\end{remark}

By Lemma \ref{F_increasing_lemma}, the equations (\ref{bn_eqn}) can thus be solved uniquely for $b_n, n=1,\dots,N$ on the interval $\left(- \sqrt{\widehat{\textrm{V}y}/\sigma_u^2},  \sqrt{\widehat{\textrm{V}y}/\sigma_u^2}\right)$. These calculations are also carried out during the odd time slots $t=2j-1, j=1,2,\dots$.

In the proposed scheme, we will not actually solve the nonlinear equations (\ref{bn_eqn}) exactly at every iteration, which is computationally intensive. Instead, since $b_1, \dots, b_N$ are constant, we will update the estimates recursively using  a stochastic approximation approach, namely
\begin{align*}
&\hat{b}_{n,j+1} = \hat{b}_{n,j} + \alpha_j \Bigg( \mathds{1}(u_{2j-1-n} > \mu) \mathds{1}(y_{2j-1} > c_{y,j}) \\ &- \int_{\mu}^\infty \Bigg[1 - \Phi \Bigg( \frac{- \hat{b}_{n,j}(u - \mu)}{\sqrt{\widehat{\textrm{V}y}_j - \hat{b}_{n,j}^2 \sigma_u^2}}\Bigg) \Bigg] p(u) du \Bigg), n=1,\dots,N,
\end{align*}
where $\widehat{\textrm{V}y}_j=(\tilde{c}_{y,j}-c_{y,j})^2$. This approach requires numerical computation of $N$ integrals (one for each $n$) at every iteration, rather than having to solve $N$ nonlinear equations (\ref{bn_eqn}) at every iteration.

In addition,  to ensure boundedness of the iterates and prove the convergence of our scheme, we will also use the idea of expanding truncations for the iterates \cite{Chen_SA}. Let $\{ M_j\}$ be a sequence of positive numbers increasing to infinity. A  recursive procedure with expanding truncations has the form
\begin{equation}
\label{expanding_truncations}
x_{j+1} = \Pi_{M_{\varsigma(j)}} (x_j + \alpha_j z_j)
\end{equation}
where
\begin{equation}
\label{increasing_thresholds}
\varsigma(0) = 0, \quad \varsigma(j) \triangleq \sum_{i=1}^{j-1} \mathds{1}(||x_i + \alpha_i z_i|| > M_{\varsigma(i)}), 
\end{equation}
and the truncation operation
\begin{equation}
\label{truncation_operation}
\Pi_M(x) \triangleq \left\{\begin{array}{ll} x, & ||x|| \leq M \\ x^*, &||x|| > M. \end{array}  \right.
\end{equation}
Thus the procedure (\ref{expanding_truncations})  truncates the iterate $x_{j+1}$ back to $x^*$ when its norm exceeds a threshold $M_{\varsigma(j)}$, with the threshold increasing each time it is exceeded, according to (\ref{increasing_thresholds}). In this paper we will choose $x^* = \mathbf{0}$. 
As $M_j$ goes to infinity, the iterates will eventually almost surely have norm less than $M_j$ for a sufficiently large $M_j$, provided conditions such as those in Theorem 2.4.1 of \cite{Chen_SA} (which we will verify as part of the proof of Theorem \ref{dumb_quantizer_convergence_thm}) are satisfied. 

Now that the intuitive ideas have been presented,  the identification scheme is formally stated as Algorithm 1 below. 

\begin{table*}[b!]
\fbox{\parbox{\textwidth}{
\textbf{Algorithm 1}
\begin{itemize} 
\item Set $c_u = \mu$, and choose a sequence $\{\alpha_j\}$  satisfying $\alpha_j > 0$, $\sum_{j=0}^\infty \alpha_j = \infty$,  and $\sum_{j=0}^\infty \alpha_j^2 < \infty$
\item Initialize $c_{y,1}=0$, $\tilde{c}_{y,1}=1$, $\hat{b}_{n,1}=0, n=1,\dots,N$
\item For $j=1,2,\dots$, compute:
\begin{equation}
\label{alg1_recusions}
\begin{split}
& \left[ \!\!\!\begin{array}{c} c_{y,j+1} \\ \tilde{c}_{y,j+1} \\ \hat{b}_{1,j+1} \\ \vdots \\ \hat{b}_{N,j+1}  \end{array} \!\!\! \right]  = \Pi_{M_{\varsigma(j)}}\left( \left[\!\!\! \begin{array}{c} c_{y,j} \\ \tilde{c}_{y,j} \\ \hat{b}_{1,j} \\ \vdots \\ \hat{b}_{N,j} \end{array} \!\!\! \right] + \alpha_j \left[\!\!\!\begin{array}{c} \mathds{1} (y_{2j-1} > c_{y,j}) - 0.5 \\  \mathds{1} (y_{2j} > \tilde{c}_{y,j}) - 0.1587 \\ \mathds{1}(u_{2j-2} > \mu) \mathds{1}(y_{2j-1} > c_{y,j}) - F(\hat{b}_{1,j},\widehat{\textrm{V}y}_j) \\ \vdots \\  \mathds{1}(u_{2j-1-N} > \mu) \mathds{1}(y_{2j-1} > c_{y,j}) - F(\hat{b}_{N,j},\widehat{\textrm{V}y}_j) \end{array} \!\!\! \right] \right)
\end{split}
\end{equation}
where $\Pi_{M_{\varsigma(j)}}(\cdot)$ is defined by (\ref{expanding_truncations})-(\ref{truncation_operation}), $F(\cdot,\cdot)$ by (\ref{F_fn}),
and $\widehat{\textrm{V}y}_{j} \triangleq (\tilde{c}_{y,j} - c_{y,j})^2$
\end{itemize}}}
\end{table*}

In Algorithm 1, note that  
 the integrals $F(\hat{b}_{n,j},\widehat{\textrm{V}y}_j)$ for $n=1,\dots,N$ can be evaluated by lookup table, by precomputing $\int_{\mu}^\infty \left[1 - \Phi  \left( - x(u - \mu)\right) \right] p(u) du$ for different values of $x$, which can substantially improve the running time of the algorithm. 

We will now prove the strong consistency of Algorithm 1.
\begin{theorem}
\label{dumb_quantizer_convergence_thm}
Under Algorithm 1 and Assumptions $1-4$, $\hat{b}_{n,j}  \stackrel{\textrm{a.s.}}{\rightarrow} b_n$ as $j \rightarrow \infty$ for $n=1,\dots,N$. 
\end{theorem}

\begin{proof}
See Appendix \ref{dumb_quantizer_convergence_thm_proof}. 

\end{proof}

\subsection{Unknown Parameters of Input Distribution}
\label{dumb_quantizer_unknown_input_sec}
In this subsection, we will relax Assumption 4, and assume that $\mu$ and $\sigma^2_u$ are also unknown.  The idea in the scheme below (Algorithm 2) is to estimate these quantities in a similar manner to how $\mathbb{E}[y]$ and $\textrm{Var}[y]$ were estimated in Algorithm~1.

However, a complication arises if we also try to estimate $\mathbb{E}[u]$ during the odd time slots and estimate  $\textrm{Var}[u]$ during the even time slots (or vice versa). This is because some of the quantities $\mathds{1}(u_{t-n} > c_{u}) \mathds{1}(y_t > c_{y}), n=1,\dots,N$, which are used in updating the parameter estimates, cannot be constructed at the estimator since we only have $\mathds{1}(u_{\tau} > c_{u})$ when $\tau$ is odd. 

To get around this difficulty, we propose the following. We will continue to estimate $\mathbb{E}[y]$  during the odd time slots $1,3,5,\dots$, and  to estimate $\textrm{Var}[y]$ during the even time slots $2,4,6,\dots$. But we will estimate  $\mathbb{E}[u]$ at time slots $1,2,5,6,9,10,\dots$, i.e. $2(j-1)+[j]_2, j=1,2,\dots$ where%\footnote{More generally, $[j]_p$ represents the congruence class of $j$ modulo $p$.} 
\begin{equation}
\label{residue_class}
[j]_2 \triangleq \left\{ \begin{array}{cl} 0, &  j \equiv 0  \textrm{ (mod }2)\\ 1, & j \equiv 1 \textrm{ (mod }2), \end{array} \right. 
\end{equation}
 and we will estimate $\textrm{Var}[u]$ at time slots $3,4,7,8,11,12,\dots,$ i.e.  $2j+[j]_2, j=1,2,\dots$. Then there will be sufficient overlap to construct the quantities $\mathds{1}(u_{t-n} > c_{u}) \mathds{1}(y_t > c_{y})$. In order to see this, note that the odd time slots have the form of either $4(k-1)+1$ or $4(k-1)+3$ for $k=1,2,\dots$, while the time slots $1,2,5,6,9,10,\dots$ have the form of either $4(k'-1)+1$ or $4(k'-1)+2$ for $k'=1,2,\dots$. So the estimator can construct the quantities $\mathds{1}(u_{t'} > c_{u}) \mathds{1}(y_t > c_{y})$, for $t=4(k-1)+1$ or  $t=4(k-1)+3$, and $t'=4(k'-1)+1$ or $t'=4(k'-1)+2$. We have the following result:
\begin{lemma}
\label{dumb_quantizer_unknown_input_lemma}
Let $t$ be either of the form  $t=4(k-1)+1$ or  $t=4(k-1)+3$, and let $t'$ be either of the form $t'=4(k'-1)+1$ or $t'=4(k'-1)+2$. Then for any $n \in \{1,\dots,N\}$, there are infinitely many pairs $(k,k') \in \mathbb{N} \times \mathbb{N}$ satisfying
$$t-t' = n.$$
\end{lemma}
\begin{proof}
For each of the different forms of $t$ and $t'$, we have $t-t'$ given by 
\begin{equation}
\label{t_difference_mod4}
\begin{split}
4(k-1)+1 - [4(k'-1)+1 ] & = 4 k - 4k'  \equiv  0 \textrm{ (mod } 4)\\
\textrm{or } 4(k-1)+1 - [4(k'-1)+2 ] & = 4 k - 4k' -1  \equiv 3 \textrm{ (mod } 4)\\
\textrm{or } 4(k-1)+3 - [4(k'-1)+1 ] & = 4 k - 4k' +2  \equiv 2 \textrm{ (mod } 4)\\
\textrm{or } 4(k-1)+3 - [4(k'-1)+2 ] & = 4 k - 4k' +1  \equiv 1 \textrm{ (mod } 4)
\end{split}
\end{equation}
Now any $n \in \{1,\dots,N\}$ must be equal to one of 0, 1, 2, or 3 modulo 4. 
Suppose first that $n \equiv 1 \textrm{ (mod } 4)$. Pick an arbitrary $k \in \mathbb{N}$. Then for $t$ of the form $t=4(k-1)+3$, and $t'$ of the form  $t'=4(k'-1)+2$, we have from the last line of (\ref{t_difference_mod4}) that $t-t' = n$ is satisfied for $k' = \frac{4k+1-n}{4}$, and $ k' \in \mathbb{N}$ since $n \equiv 1 \textrm{ (mod } 4)$. As $k$ is arbitrary, one can find infinitely many pairs $(k,k') \in \mathbb{N} \times \mathbb{N}$ satisfying  $t-t' = n$ when  $n \equiv 1 \textrm{ (mod } 4)$.  

A similar argument applies when  $n$  modulo 4 is equal to 0, 2, or 3. 
\end{proof}

The identification scheme is formally given  as Algorithm 2, where we use the variables:
\begin{equation}
\label{g_j_bar_defn}
\begin{split}
g(n,j)  &\triangleq \! \left\{ \begin{array}{cl} \mathds{1}\big(2j-1 \!\equiv\! 1 \!\textrm{ (mod } \!4) \big), & \textrm{if } n \equiv 0 \textrm{ (mod } 4) \\ & \textrm{or } n \equiv 3 \textrm{ (mod } 4) \\
\mathds{1}\big(2j-1 \!\equiv\! 3 \!\textrm{ (mod } \!4) \big), & \textrm{if } n \equiv 1 \textrm{ (mod } 4) \\ & \textrm{or } n \equiv 2 \textrm{ (mod } 4),
 \end{array} \right. \\
 \bar{j}(n,j) &  \triangleq \left\{ \begin{array}{cl}  j - \lfloor \frac{n}{2} \rfloor, & \textrm{if }  j - \lfloor \frac{n}{2} \rfloor \geq 1 \\ 1, & \textrm{otherwise},  \end{array} \right. 
 \end{split}
\end{equation}
to keep track of which parameters can be updated and past thresholds. In addition we also use the function $h(c_y,\tilde{c}_y,c_u,\tilde{c}_u,b)$ defined by (\ref{h_defn}).
\begin{figure*}[!t]
\begin{align}
\label{h_defn}
& h(c_y,\tilde{c}_y,c_u,\tilde{c}_u,b) \nonumber \\
& \triangleq \left\{ \begin{array}{ll} \int\limits_{c_u}^\infty \! \Big[1 \!-\! \Phi  \Big( \frac{- b(u - c_u)}{\sqrt{(\tilde{c}_y - c_y)^2 - b^2 (\tilde{c}_u - c_u)^2}}\Big) \Big]  \frac{1}{\sqrt{2 \pi (\tilde{c}_u - c_u)^2} } \textrm{exp} \left(- \frac{(u - c_u)^2}{2 (\tilde{c}_u - c_u)^2} \right)du, & - \sqrt{\frac{(\tilde{c}_y-c_y)^2}{(\tilde{c}_u-c_u)^2}} < b < \sqrt{\frac{(\tilde{c}_y-c_y)^2}{(\tilde{c}_u-c_u)^2}} \\
 \int\limits_{c_u}^\infty \! \frac{1}{\sqrt{2 \pi (\tilde{c}_u - c_u)^2} } \textrm{exp} \left(- \frac{(u - c_u)^2}{2 (\tilde{c}_u - c_u)^2} \right)du  = \frac{1}{2},  & b \geq \sqrt{\frac{(\tilde{c}_y-c_y)^2}{(\tilde{c}_u-c_u)^2}}\\ 
 0, & b \leq -\sqrt{\frac{(\tilde{c}_y-c_y)^2}{(\tilde{c}_u-c_u)^2}}. \end{array} \right.
\end{align}
\end{figure*}

\begin{table*}[t!]
\fbox{\parbox{\textwidth}{
\textbf{Algorithm 2}
\begin{itemize} 
\item Choose a sequence $\{\alpha_j\}$  satisfying $\alpha_j > 0$, $\sum_{j=0}^\infty \alpha_j = \infty$,  and $\sum_{j=0}^\infty \alpha_j^2 < \infty$
\item Initialize $c_{y,1}=0$, $\tilde{c}_{y,1}=1$, $c_{u,1}=0$, $\tilde{c}_{u,1}=1$, $\hat{b}_{n,1}=0, n=1,\dots,N$
\item For $j=1,2,\dots$,  compute:
\begin{equation}
\label{alg2_recursions}
\begin{split}
& \left[ \!\!\!\begin{array}{c} c_{y,j+1} \\ \tilde{c}_{y,j+1} \\ c_{u,j+1} \\ \tilde{c}_{u,j+1} \\ \hat{b}_{1,j+1} \\ \vdots \\ \hat{b}_{N,j+1}  \end{array} \!\!\! \right]  = \Pi_{M_{\varsigma(j)}}\left( \left[\!\!\! \begin{array}{c} c_{y,j} \\ \tilde{c}_{y,j} \\  c_{u,j} \\ \tilde{c}_{u,j} \\ \hat{b}_{1,j} \\ \vdots \\ \hat{b}_{N,j} \end{array} \!\!\! \right] + \alpha_j \left[\!\!\!\begin{array}{c} \mathds{1} (y_{2j-1} > c_{y,j}) - 0.5 \\  \mathds{1} (y_{2j} > \tilde{c}_{y,j}) - 0.1587 \\  \mathds{1} (u_{2(j-1)+[j]_2} > c_{u,j}) - 0.5 \\  \mathds{1} (u_{2j+[j]_2} > \tilde{c}_{u,j}) - 0.1587 \\ 
g(1,j) \Big[\mathds{1} (u_{2j-2} > c_{u,\bar{j}(1,j)}) \mathds{1} (y_{2j-1} > c_{y,j}) %\\ \quad\quad 
- h(c_{y,j}, \tilde{c}_{y,j},c_{u,\bar{j}(1,j)}, \tilde{c}_{u,j}, \hat{b}_{1,j} \Big]  \\ \vdots \\  
g(N,j) \Big[\mathds{1} (u_{2j-1-N} > c_{u,\bar{j}(N,j)}) \mathds{1} (y_{2j-1} > c_{y,j}) %\\ \quad\quad 
- h(c_{y,j}, \tilde{c}_{y,j},c_{u,\bar{j}(N,j)}, \tilde{c}_{u,j}, \hat{b}_{N,j} \Big] \end{array} \!\!\! \right] \right)
\end{split}
\end{equation}
where  $\Pi_{M_{\varsigma(j)}}(\cdot)$ is defined by (\ref{expanding_truncations})-(\ref{truncation_operation}), $[\cdot]_2$  by (\ref{residue_class}), $g(\cdot,\cdot)$ and $\bar{j}(\cdot,\cdot)$ by (\ref{g_j_bar_defn}), and $h(\cdot,\cdot,\cdot,\cdot,\cdot)$ by (\ref{h_defn}).
\end{itemize}}}
\end{table*}
By Lemma \ref{dumb_quantizer_unknown_input_lemma}, there will  be an infinite number of regularly spaced time slots where the quantities $\mathds{1}(u_{t-n} > c_{u}) \mathds{1}(y_t > c_{y})$ for each $n \in \{1,\dots,N\}$ can be constructed at the estimator. In particular, the different cases in the definition of $ g(n,j)$ in Algorithm 2 follow from (\ref{t_difference_mod4}) in the proof of Lemma \ref{dumb_quantizer_unknown_input_lemma}. Note also that the integral 
\begin{align*}
&\int_{c_u}^\infty \Big[1 - \Phi  \Big( \frac{- b(u - c_u)}{\sqrt{(\tilde{c}_y - c_y)^2 - b^2 (\tilde{c}_u - c_u)^2}}\Big) \Big] \\ & \qquad\times\frac{1}{\sqrt{2 \pi (\tilde{c}_u - c_u)^2} } \textrm{exp} \left(- \frac{(u - c_u)^2}{2 (\tilde{c}_u - c_u)^2} \right)du 
\end{align*}
in (\ref{h_defn}) can be evaluated by lookup table, by using a change of variable $v = u-c_u$ and precomputing \\
$\int_{0}^\infty \left[1 - \Phi  \left( -x v\right) \right] \frac{1}{\sqrt{2 \pi z} } \textrm{exp} \left(- \frac{v^2}{2 z} \right)dv $ for different values of $x$ and $z$. 

\begin{theorem}
\label{dumb_quantizer_convergence_thm2}
Under Algorithm 2 and Assumptions $1-3$, $\hat{b}_{n,j}  \stackrel{\textrm{a.s.}}{\rightarrow} b_n$ as $j \rightarrow \infty$ for $n=1,\dots,N$. 
\end{theorem}

\begin{proof}
See Appendix \ref{dumb_quantizer_convergence_thm2_proof}.
\end{proof}

\subsection{Simulation Results}
\label{dumb_quantizer_sim_sec}
We consider a third order system with $\mu=1$, $\sigma_u^2=1$, $\sigma_w^2=1$, $b_1=0.2$, $b_2=-0.2$, $b_3=0.6$. 
In the plots below  we will use the sequence $\alpha_j = \frac{10}{j}$. An initial truncation bound of $M_0=1000$ was used, but was never exceeded in our simulations. 
We first consider the case where  $\mu$ and $\sigma_u^2$ are known to the estimator. Fig. \ref{dumb_quantizer_alg1_plot} shows the estimates $\hat{b}_1, \hat{b}_2, \hat{b}_3$ from Algorithm 1,  and as expected from Theorem \ref{dumb_quantizer_convergence_thm}, they converge to the true values.
\begin{figure}[t!]
\centering 
\includegraphics[scale=0.5]{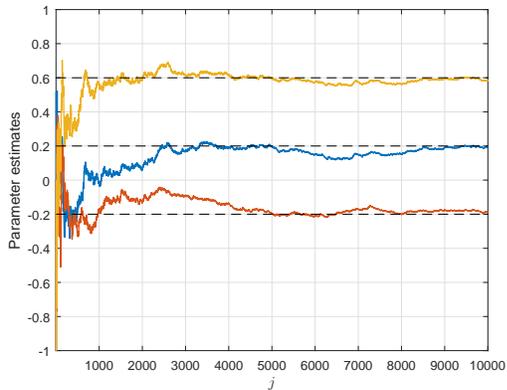} 
\caption{Parameter Estimates: Algorithm 1}
\label{dumb_quantizer_alg1_plot}
\end{figure}

To look at the convergence behaviour, we can approximate the variance of $j^{1/2} (\hat{b}_{n,j}-b_n), n=1,2,3$  \cite{Chen_SA}, \cite{KushnerYin}. However, in order to allow for a fairer comparison with the algorithms of Section \ref{smart_quantizer_sec}, we will instead approximate the variance of $t^{1/2} (\hat{b}_{n,j}-b_n), n=1,2,3$, where $t$ is the time index. This is done by computing the sample variance over 10000 different  simulation runs of Algorithm 1, and are given in Fig.  \ref{dumb_quantizer_alg1_convergence_plot}. 
\begin{figure}[t!]
\centering 
\includegraphics[scale=0.5]{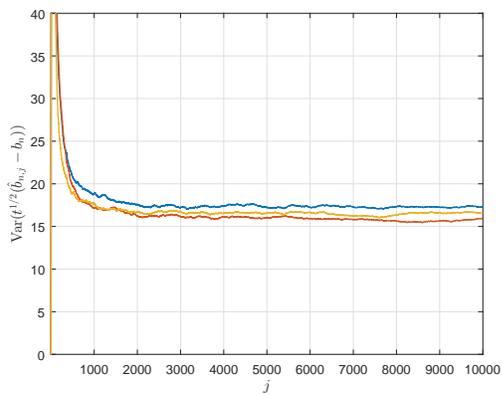} 
\caption{Convergence Behaviour: Algorithm 1}
\label{dumb_quantizer_alg1_convergence_plot}
\end{figure}

Next, we consider the system identification scheme of Section \ref{dumb_quantizer_unknown_input_sec} where $\mu$ and $\sigma_u^2$ are not assumed to be known. 
Fig. \ref{dumb_quantizer_alg2_plot} shows the estimates $\hat{b}_1, \hat{b}_2, \hat{b}_3$ from Algorithm 2. Also in this case, the estimates converge to the true values, in agreement with Theorem \ref{dumb_quantizer_convergence_thm2}.
\begin{figure}[t!]
\centering 
\includegraphics[scale=0.5]{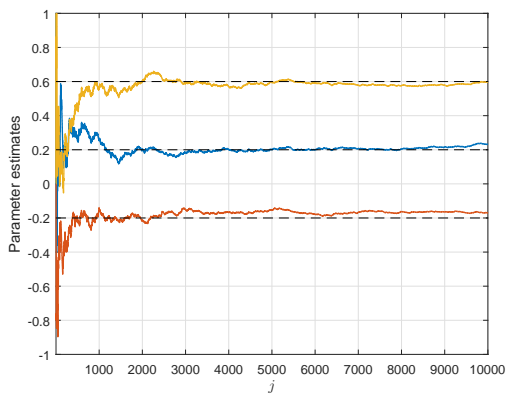} 
\caption{Parameter Estimates: Algorithm 2}
\label{dumb_quantizer_alg2_plot}
\end{figure} 
Approximations of the variances $t^{1/2} (\hat{b}_{n,j}-b_n), n=1,2,3$ using Monte Carlo approximations over 10000 simulation runs are plotted in Fig. \ref{dumb_quantizer_alg2_convergence_plot}. We see that the normalized variances  in Fig. \ref{dumb_quantizer_alg2_convergence_plot} are significantly higher (more than double) than for Algorithm 1,  due to the need to also estimate the parameters of the input distribution. 
\begin{figure}[t!]
\centering 
\includegraphics[scale=0.5]{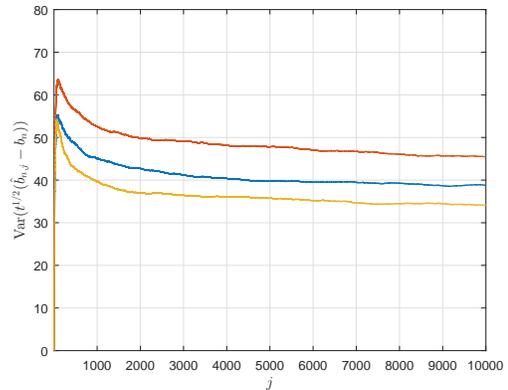} 
\caption{Convergence Behaviour: Algorithm 2}
\label{dumb_quantizer_alg2_convergence_plot}
\end{figure} 

\section{Quantizers With Computational Capabilities}
\label{smart_quantizer_sec}
The setup in Section \ref{dumb_quantizer_sec} assumes knowledge of the input and noise distributions. Specifically, we assumed that the input and noise were both Gaussian. For unknown distributions and FIR systems of order $N > 1$, it appears to be  difficult to come up with an identification scheme that is consistent and/or efficient.\footnote{For $N=1$ a consistent identification scheme was developed in \cite{LianLuo} for symmetrically distributed inputs and noises.} In this section we consider the case where the input and output quantizers are ``smart'', in the sense that they have some computational and storage capabilities, and have access to the unquantized inputs and outputs.  For instance, in many wireless sensor network applications such as in environmental monitoring \cite{PalaniswamiRaoBainbridge,YiLoMakLeung} and process industries \cite{Zhao_industrial}, the sensors used often have sensing, computation and wireless communication capabilities. In such applications the quantization or analog-to-digital (A/D) conversion is done by the sensor, and additionally these sensors would also have some on-board computing capabilities to do additional processing of the data.  For such situations  we present in this section identification schemes which can estimate the parameters for unknown input and noise distributions. 

\subsection{Data Generating System and Model}
As in Section \ref{dumb_quantizer_sec}, the system to be identified is an $N$-th order FIR system
\begin{equation}
\label{FIR_model}
y_t = b_1 u_{t-1} + b_2 u_{t-2} + \dots + b_N u_{t-N} + w_t.
\end{equation}
We now make the following assumptions:
\\ \emph{Assumption 5}: The input and output quantizers have computational and storage capabilities.
\\ \emph{Assumption 6}: The input sequence $\{u_t\}$ and the noise sequence $\{w_t\}$ are i.i.d. and mutually independent. Moreover, $w_t$ is zero mean.
\\ \emph{Assumption 7}: The model order $N$ is known.

\subsection{Identification Scheme for Known Input Distribution}
\label{smart_quantizer_scheme_sec}
In this subsection, we will also make the following assumption: 
\\ \emph{Assumption 8}: The input distribution is known to the estimator.

The noise distribution is not assumed to be known, apart from assuming that it has zero mean. As in  Section \ref{dumb_quantizer_scheme_sec}, we will start by describing the ideas involved, before formally stating the identification scheme as Algorithm 3, followed by a proof of strong consistency of the parameter estimates. 

First,  the quantized information $\mathds{1}(u_{t} > c_u), t=1,2,\dots$ sent by the input quantizer to the estimator is forwarded by the estimator to the output quantizer.
Whenever $u_{t} > c_u$, we increment an index $i$ by one. 
Denote the times where $u_{t} > c_u$ by $\tau_1, \tau_2, \dots$, with $\tau_1 < \tau_2 < \dots$. 
The output quantizer computes the following $N$  quantities (after the corresponding output is available at the quantizer):
$$d_{n,i} \triangleq \frac{1}{i} \sum_{i'=1}^i y_{\tau_{i'} + n}, \quad n=1,\dots,N,$$ using the recursions 
\begin{equation}
\label{dn_computation}
d_{n,i} = \frac{1}{i}(y_{\tau_i + n} + (i-1) d_{n,i-1}), \quad n=1,\dots,N.
\end{equation}
By  the ergodic theorem \cite[p. 393]{GrimmettStirzaker}, we have that as $i \rightarrow \infty$,
\begin{equation}
\label{d_convergence}
\begin{split}
d_{n,i} &  \stackrel{\textrm{a.s.}}{\rightarrow} \mathbb{E}[y_{t+n} | u_{t} > c_u] \\
& = \mathbb{E}\Big[\sum_{m=1}^N b_m u_{t+n-m} + n_{t + n} \Big| u_{t} > c_u\Big] \\
& = b_1  \mathbb{E}[u] + \dots + b_{n-1}  \mathbb{E}[u] + b_n \mathbb{E}[u | u > c] \\ & \quad + b_{n+1}  \mathbb{E}[u] + \dots + b_N \mathbb{E}[u] \\
& \triangleq d_n, \quad n=1,\dots,N.
\end{split}
\end{equation}

$d_{n,i}$ is computed at the output quantizer. In order for the estimator to be able to approximate $d_{n,i}$, information is sent from the output quantizer to the estimator as follows. Whenever the index $i$ is a multiple of $N$, 
another iteration index $j$ is incremented by one and the following estimates of $d_1,\dots,d_N$ are computed at the output quantizer:
\begin{equation}
\label{sign_algorithm}
\begin{split}
\hat{d}_{n,j+1} &= \hat{d}_{n,j} + \alpha_j \textrm{sgn}(d_{n,i} - \hat{d}_{n,j}), \quad n=1,\dots,N,
\end{split}
\end{equation}
where $\{\alpha_j\}$ is a sequence satisfying $\alpha_j > 0, \sum_{j=0}^\infty \alpha_j = \infty, \sum_{j=0}^\infty \alpha_j^2 < \infty$, and 
$$\textrm{sgn}(x) \triangleq \left\{ \begin{array}{rcl}  -1 & , & x < 0 \\ 1 & , & x > 0. \end{array} \right.$$ 
The term $\textrm{sgn}(d_{n,i}-\hat{d}_{n,j})$ is essentially binary, and is sent to the estimator by the output quantizer, which also computes $\hat{d}_{n,j+1}$ according to (\ref{sign_algorithm}), assuming that both the estimator and the quantizer have access to the initial condition $\hat{d}_{n,0}$. Alternatively, $\hat{d}_{n,j}$ can be computed at the estimator only and transmitted to the output quantizer.

Note that $\hat{d}_{n,j}$ is updated at $1/N$-th the rate of $d_{n,i}$, in order for each of the quantities $\textrm{sgn}(d_{n,i} - \hat{d}_{n,j}), n=1,\dots, N$ to be sent in separate time slots. From  (\ref{d_convergence}) and (\ref{sign_algorithm}),  we can show (see the proof of Theorem \ref{smart_quantizer_convergence_thm}) that 
\begin{equation}
\label{dhat_convergence}
\begin{split}
\hat{d}_{n,j}  \stackrel{\textrm{a.s.}}{\rightarrow} d_n & = b_1  \mathbb{E}[u] + \dots + b_{n-1}  \mathbb{E}[u] + b_n \mathbb{E}[u | u > c] \\ & \quad + b_{n+1}  \mathbb{E}[u] + \dots + b_N \mathbb{E}[u].
\end{split}
\end{equation}

Finally, the parameters $b_1,b_2,\dots,b_N$ of the $N$-th order system (\ref{FIR_model}) are estimated by solving for $\hat{b}_{1,j}, \hat{b}_{2,j}, \dots, \hat{b}_{N,j}$ the following set of linear equations:
\begin{equation}
\label{linear_eqns}
\mathbf{U}\left[\hat{b}_{1,j}, \dots,  \hat{b}_{N,j}\right]^T = \left[\hat{d}_{1,j}, \dots,  \hat{d}_{N,j}\right]^T,
\end{equation}
where
\begin{equation}
\label{U_defn} 
\mathbf{U} \triangleq \left[\begin{array}{cccc} \mathbb{E}[u|u > c_u] & \mathbb{E}[u] & \dots & \mathbb{E}[u] \\ \mathbb{E}[u] & \mathbb{E}[u|u > c_u] & \dots & \mathbb{E}[u] \\
\vdots & \vdots & \ddots & \vdots \\ \mathbb{E}[u] & \mathbb{E}[u] & \dots & \mathbb{E}[u|u>c_u] 
 \end{array} \right]. 
 \end{equation} 
Note that $\mathbf{U}$ is known at the estimator, since by Assumption 8 the estimator  knows the input distribution. 
The equations (\ref{linear_eqns}) will have a unique solution under the following assumption:
\\ \emph{Assumption 9}: The input distribution of $u$ and input quantizer threshold $c_u$ satisfies $\mathbb{E}[u|u > c_u] \neq \mathbb{E}[u]$ and $\mathbb{E}[u|u > c_u] \neq (1-N) \mathbb{E}[u]$. 

We note that apart from degenerate cases such as $u$ being constant, $c_u$ can always be chosen such that Assumption 9 is satisfied.
Under Assumption 9, when $\mathbb{E}[u] \neq 0$ we have uniqueness of solutions to (\ref{linear_eqns}) by the following result:
\begin{lemma}
\label{invertible_matrix_lemma}
The $N \times N$ matrix 
\begin{equation}
\label{A_matrix}
\mathbf{A} \triangleq \left[ \begin{array}{cccc} a & 1 & \dots & 1 \\ 1 & a & \dots & 1 \\ \vdots & \vdots & \ddots & \vdots \\ 1 & 1 & \dots & a \end{array} \right]
\end{equation}
is invertible if $a \neq 1 $ and $a \neq 1-N$. 
\end{lemma}

\begin{proof}
We use the property that a matrix $\mathbf{A}$ is invertible if and only if $\mathbf{A x} =\mathbf{0} \Rightarrow \mathbf{x} = \mathbf{0}$. Denoting $\mathbf{x} \triangleq [x_1,x_2,\dots,x_N]^T$, $\mathbf{Ax}=\mathbf{0}$ (for $\mathbf{A}$ given by (\ref{A_matrix})) is equivalent to 
\begin{equation}
\label{linear_eqns2}
\begin{split}
&a x_1 + x_2 + \dots + x_N  = 0 \\
&x _1 + a x_2 + \dots + x_N  = 0 \\
&\quad\quad\quad \vdots \\
&x_1 + x_2 + \dots + a x_N  = 0.
\end{split}
\end{equation}
Subtracting the second equation from the first equation  in (\ref{linear_eqns2}), we have $(x_1-x_2)(a-1) = 0$, which implies that $x_1=x_2$ since $a \neq 1$. Repeating this  argument leads to
\begin{equation}
\label{xn_equal}
x_1=x_2=\dots=x_N.
\end{equation} 
Using (\ref{xn_equal}) on the first equation of (\ref{linear_eqns2}), we have 
$(a+N-1)x_1 =0$, which implies $x_1=0$ since $a \neq 1-N$. Hence $\mathbf{x}=\mathbf{0}$. 
\end{proof}

We now formally state the identification scheme as Algorithm 3. In the formal description, the sets $\mathcal{D}_t$ and indices $i_n$ are used to keep track of which of the quantities $d_{n,i}, \ n=1,\ldots,N$, should be updated at time $t$.

\begin{table*}[t!]
\fbox{\parbox{\textwidth}{
\textbf{Algorithm 3}
\begin{itemize} 
\item Choose a $c_u$ satisfying Assumption 9, and a sequence $\{\alpha_j\}$ satisfying $\alpha_j > 0$, $ \sum_{j=0}^\infty \alpha_j = \infty$, and $\sum_{j=0}^\infty \alpha_j^2 < \infty$ 
\item Initialize $i=0, j=0$,  $d_{n,0} = 0, \hat{d}_{n,0}=0, i_n = 1, n=1,\dots,N, \quad \mathcal{D}_t = \emptyset, \forall t$
%\item Initialize $d_{n,0} = 0, n=1,\dots,N$ at the output quantizer 
%\item Initialize $\hat{d}_{n,0} = 0, n=1,\dots,N$ at both the output quantizer and estimator. 
\item For  $t=1,2,\dots$, do:
\begin{itemize}
\item If $u_{t} > c_u$, set $\tau_i = t$, $i := i + 1$
\begin{itemize} \item[--] If $ i \equiv 0 \textrm{ (mod } N)$, set $j := j + 1$ \end{itemize}
\item At the input quantizer:
\begin{enumerate}
\item Send $\mathds{1}(u_{t} > c_u)$ to estimator, which passes it on to the output quantizer 
\end{enumerate}
\item At the output quantizer, when $u_t > c_u$:
\begin{enumerate}
\item Set $\mathcal{D}_{t+n} := \mathcal{D}_{t+n} \bigcup \{n\}, n=1,\dots,N$
\item Compute $d_{n,i_n} = \frac{1}{i_n} (y_{t} + (i_n-1) d_{n,i_n-1})$ and set $i_n := i_n +1$ for all $n \in \mathcal{D}_t$, and remove $\mathcal{D}_{t-1}$ from memory
\item When $ i \equiv 0 \textrm{ (mod } N)$,  compute $\textrm{sgn}(d_{n,i} - \hat{d}_{n,j})$ and $\hat{d}_{n,j+1} = \hat{d}_{n,j} + \alpha_j \textrm{sgn}(d_{n,i} - \hat{d}_{n,j})$ at time $\tau_i+n$ for  $n=1,\dots,N$. Send $\textrm{sgn}(d_{n,i} - \hat{d}_{n,j})$ at time $\tau_i+n$ to estimator, for  $n=1,\dots,N$
\end{enumerate}
\item At the estimator, when $ i \equiv 0 \textrm{ (mod } N)$:
\begin{enumerate}
%\item Send $\mathds{1}(u_{t} > c_u)$ to output quantizer
\item Compute $\left[\hat{b}_{1,j}, \dots,  \hat{b}_{N,j}\right]^T = \mathbf{U}^{-1} \left[\hat{d}_{1,j}, \dots,  \hat{d}_{N,j}\right]^T$, where $ \mathbf{U}$ is defined by (\ref{U_defn})
\item Compute $\hat{d}_{n,j+1} = \hat{d}_{n,j} + \alpha_j \textrm{sgn}(d_{n,i} - \hat{d}_{n,j})$ when $\textrm{sgn}(d_{n,i} - \hat{d}_{n,j})$ arrives at estimator, for $n=1,\dots,N$ 
\end{enumerate}
\end{itemize}
\end{itemize}
}}
\end{table*}

\begin{theorem}
\label{smart_quantizer_convergence_thm}
Under Algorithm 3 and Assumptions $5-9$, $\hat{b}_{n,j}  \stackrel{\textrm{a.s.}}{\rightarrow} b_n$ as $j \rightarrow \infty$ for $n=1,\dots,N$. 
\end{theorem}
\begin{proof}
See Appendix \ref{smart_quantizer_convergence_thm_proof}.
\end{proof}

\subsection{Unknown Input Distribution}
\label{smart_quantizer_unknown_input_sec}
Solving the linear equations (\ref{linear_eqns}) requires knowledge of $\mathbb{E}[u]$ and $\mathbb{E}[u|u > c_u]$, which in turn requires knowledge of the distribution of $u$. When the input distribution is unknown (except for enough knowledge such that Assumption 9 can be satisfied),  $\mathbb{E}[u]$ and $\mathbb{E}[u|u > c_u]$ can be estimated  if we also allow for some computation at the input quantizer. 

To estimate $\mathbb{E}[u]$, the input quantizer first computes 
$$e_{1,t} \triangleq \frac{1}{t} \sum_{t'=1}^t u_{t'}$$
using the recursion 
$$e_{1,t} = \frac{1}{t}(u_t +(t-1) e_{1,t-1}).$$
By the strong law of large numbers, $e_{1,t} \stackrel{\textrm{a.s.}}{\rightarrow} \mathbb{E}[u]$ as $t \rightarrow \infty$. 
The estimator estimates $e_{1,t}$ using the recursion: 
$$\hat{e}_{1,j+1} = \hat{e}_{1,j} + \alpha_j \textrm{sgn}(e_{1,t} - \hat{e}_{1,j})$$
where the quantities $ \textrm{sgn}(e_{1,t} - \hat{e}_{1,j})$  are sent by the input quantizer (see below for how the index $j$ is updated). Again, $\{\hat{e}_{1,j}\}$ can be reconstructed at the input quantizer given knowledge of the initial condition $\hat{e}_{1,0}$. 

To estimate $\mathbb{E}[u|u>c_u]$, whenever $u_t > c_u$, the input quantizer first increments an index $k$ by one. Denote the times when $u_t > c_u$ by $t_1, t_2, \dots$, with $t_1 < t_2 < \dots$. The input quantizer then computes
$$ e_{2,k} \triangleq \frac{1}{k} \sum_{k'=1}^k u_{t_{k'}}$$ 
using the recursion
$$e_{2,k} = \frac{1}{k} (u_{t_k} +(k-1) e_{2,k-1}).$$
We have  $e_{2,k} \stackrel{\textrm{a.s.}}{\rightarrow} \mathbb{E}[u|u>c_u]$ as $k \rightarrow \infty$ by the strong law of large numbers. The estimator estimates $e_{2,k}$ using the recursion:
$$\hat{e}_{2,j+1} = \hat{e}_{2,j} + \alpha_j \textrm{sgn}(e_{2,k} - \hat{e}_{2,j})$$ 
where the quantities $ \textrm{sgn}(e_{2,k} - \hat{e}_{2,j})$ are sent by the input quantizer. 

Now in Algorithm 3, the input quantizer is already sending $\mathds{1}(u_t > c_u)$ to the estimator at every time slot. Thus we need to modify the division of the time slots to incorporate the sending of the additional information $\textrm{sgn}(e_{1,t}-\hat{e}_{1,j})$ and $\textrm{sgn}(e_{2,k}-\hat{e}_{2,j})$ . 
We propose the following: Instead of an iteration $j$ having a (minimum) length of $N$ time slots as in Algorithm 3, we will now consider iterations $j$ with a (minimum) length of $N+2$ time slots. During the first $N$ time slots, the input quantizer will send  $\mathds{1}(u_t > c_u)$ to the estimator, which are then forwarded to the output quantizer. As in Algorithm 3, an index $i$ is now incremented by one\footnote{The indices $i$ and $k$ are different, as in the updating of $k$ one checks if $u_t > c_u$ at \emph{every} time step, to obtain more accurate estimates.} every time  $u_t > c_u$ (during the first $N$ time slots), and the iteration index $j$ is incremented by one whenever $i$ is a multiple of $N$.  The remaining two time slots will be used to transmit the quantities $\textrm{sgn}(e_{1,t} - \hat{e}_{1,j})$ and $\textrm{sgn}(e_{2,k} - \hat{e}_{2,j})$.

The parameters $b_1,b_2,\dots,b_N$  are now estimated by solving for $\hat{b}_{1,j}, \hat{b}_{2,j}, \dots, \hat{b}_{N,j}$ the following set of linear equations:
\begin{equation}
\label{linear_eqns_unknown_input}
\mathbf{U}_j \left[\hat{b}_{1,j}, \dots,  \hat{b}_{N,j}\right]^T = \left[\hat{d}_{1,j}, \dots,  \hat{d}_{N,j}\right]^T
\end{equation}
where
\begin{equation}
\label{Uj_defn}
\mathbf{U}_j \triangleq  \left[\begin{array}{cccc} \hat{e}_{2,j} & \hat{e}_{1,j} & \dots & \hat{e}_{1,j} \\ \hat{e}_{1,j} & \hat{e}_{2,j} & \dots & \hat{e}_{1,j} \\
\vdots & \vdots & \ddots & \vdots \\ \hat{e}_{1,j} & \hat{e}_{1,j} & \dots & \hat{e}_{2,j}
 \end{array} \right].
\end{equation}

The formal statement of the identification scheme is given  as Algorithm 4.

\begin{table*}[t!]
\fbox{\parbox{\textwidth}{
\textbf{Algorithm 4}
\begin{itemize} 
\item Choose a $c_u$ satisfying Assumption 9, and a  sequence $\{\alpha_j\}$  satisfying $\alpha_j > 0$, $\sum_{j=0}^\infty \alpha_j = \infty$, and $\sum_{j=0}^\infty \alpha_j^2 < \infty$. 
\item Initialize $i=0, j=0, k=0$, $d_{n,0}=0, \hat{d}_{n,0}=0, i_n=1, n=1,\dots,N$,  $e_{1,0}=0, e_{2,0}=0$,  $\hat{e}_{1,0} = 0, \hat{e}_{2,0} = 0 $, $\mathcal{D}_t = \emptyset, \forall t$
%\item Initialize $d_{n,0}=0, n=1,\dots,N$ at the output quantizer
%\item Initialize $\hat{d}_{n,0} = 0, n=1,\dots,N$ at both the output quantizer and estimator. 
%\item Initialize $e_{1,0}=0, e_{2,0}$ at the input quantizer
%\item Initialize $\hat{e}_{1,0} = 0, \hat{e}_{2,0} = 0 $ at both the input quantizer and estimator.
\item For $t=1,2,\dots$, do:
\begin{itemize}
\item If $u_t > c_u$, set $t_k=t$, $k := k + 1$
\item If $t \textrm{ mod } (N+2) \in \{1,\dots,N\}$ and $u_t > c_u$, set $\tau_i=t$, $i  := i + 1$
\begin{itemize} \item[--] If $i \equiv 0 \textrm{ (mod } N)$, set $j  := j + 1$ \end{itemize}
\item At the input quantizer:
\begin{enumerate}
\item Compute $e_{1,t} = \frac{1}{t} (u_t + (t-1) e_{1,t-1})$ and  $e_{2,k} = \frac{1}{k} (u_{t_k} +(k-1) e_{2,k-1})$
\item Send $\mathds{1}(u_t >c_u)$ to estimator if $t \textrm{ mod } (N+2) \in \{1,\dots,N\}$, which passes it on to the output quantizer 
\item When $i \equiv 0 \textrm{ (mod } N)$, compute $\textrm{sgn}(e_{1,t} - \hat{e}_{1,j})$, $\textrm{sgn}(e_{2,k} - \hat{e}_{2,j})$, 
$\hat{e}_{1,j+1} = \hat{e}_{1,j} + \alpha_j \textrm{sgn}(e_{1,t} - \hat{e}_{1,j})$, and $\hat{e}_{2,j+1} = \hat{e}_{2,j} + \alpha_j \textrm{sgn}(e_{2,k} - \hat{e}_{2,j})$. Send $\textrm{sgn}(e_{1,t} - \hat{e}_{1,j})$ and $\textrm{sgn}(e_{2,k} - \hat{e}_{2,j})$ to estimator at times $\tau_i + N + 1$ and $\tau_i + N + 2$ respectively
\end{enumerate}
\item At the output quantizer, when $t \textrm{ mod } (N+2) \in \{1,\dots,N\}$ and $u_t > c_u$:
\begin{enumerate}
\item Set $\mathcal{D}_{t+n} := \mathcal{D}_{t+n} \bigcup \{n\}, n=1,\dots,N$
\item Compute $d_{n,i_n} = \frac{1}{i_n} (y_t + (i_n-1) d_{n,i_n-1})$ and set $i_n := i_n +1$ for all $n \in \mathcal{D}_t$, and remove $\mathcal{D}_{t-1}$ from memory
\item When  $i \equiv 0 \textrm{ (mod } N)$, compute $\textrm{sgn}(d_{n,i} - \hat{d}_{n,j})$ and $\hat{d}_{n,j+1} = \hat{d}_{n,i} + \alpha_j \textrm{sgn}(d_{n,i} - \hat{d}_{n,j})$ at time $\tau_i +n$, for  $ n=1,\dots,N$. Send $\textrm{sgn}(d_{n,i} - \hat{d}_{n,j})$ at time $\tau_i +n$ to estimator, for $n=1,\dots,N$ 
\end{enumerate}
\item At the estimator, when $i \equiv 0 \textrm{ (mod } N)$:
\begin{enumerate}
%\item Send $\mathds{1}(u_{t} > c_u)$ for $t \in \{1  \textrm{ mod } (N+2),\dots,N  \textrm{ mod } (N+2)\}$ to output quantizer
\item Compute 
$$ \left[\hat{b}_{1,j}, \dots,  \hat{b}_{N,j}\right]^T =  \left\{ \begin{array} {ccl} \mathbf{U}_j^{-1} \left[\hat{d}_{1,j}, \dots,  \hat{d}_{N,j}\right]^T & , & \textrm{if }  \mathbf{U}_j \textrm{ is invertible} \\ 0 & , & \textrm{otherwise} \end{array} \right.$$
where $\mathbf{U}_j$ is defined by (\ref{Uj_defn})
\item Compute $\hat{d}_{n,j+1} = \hat{d}_{n,i} + \alpha_j \textrm{sgn}(d_{n,i} - \hat{d}_{n,j}),  n=1,\dots,N$, $\hat{e}_{1,j+1} = \hat{e}_{1,j} + \alpha_j \textrm{sgn}(e_{1,t} - \hat{e}_{1,j})$, and $\hat{e}_{2,j+1} = \hat{e}_{2,j} + \alpha_j \textrm{sgn}(e_{2,k} - \hat{e}_{2,j})$, when the quantities $\textrm{sgn}(d_{n,i} - \hat{d}_{n,j})$, $\textrm{sgn}(e_{1,t} - \hat{e}_{1,j})$, $\textrm{sgn}(e_{2,k} - \hat{e}_{2,j})$ arrive at estimator
\end{enumerate}
\end{itemize}
\end{itemize}
}}
\end{table*}

\begin{theorem}
Under Algorithm 4 and Assumptions $5, 6, 7, 9$, $\hat{b}_{n,j}  \stackrel{\textrm{a.s.}}{\rightarrow} b_n$ as $j \rightarrow \infty$ for $n=1,\dots,N$. 
\end{theorem}

\begin{proof}
Using similar arguments as in the  proof of Theorem \ref{smart_quantizer_convergence_thm}, we can show that 
\begin{equation}
\label{dhat_ehat_convergence}
\begin{split}
\hat{d}_{n,j} &  \stackrel{\textrm{a.s.}}{\rightarrow}  b_1  \mathbb{E}[u] + \dots + b_{n-1} \mathbb{E}[u] +  b_{n}  \mathbb{E}[u|u > c_u] \\ & \qquad + b_{n+1}  \mathbb{E}[u] + \dots + b_N \mathbb{E}[u], \quad n=1,\dots,N, \\
\hat{e}_{1,j} & \stackrel{\textrm{a.s.}}{\rightarrow} \mathbb{E}[u], \textrm{ and }
\hat{e}_{2,j}  \stackrel{\textrm{a.s.}}{\rightarrow} \mathbb{E}[u|u>c_u].
\end{split}
\end{equation}

Hence by (\ref{dhat_ehat_convergence}) and continuity,
\begin{equation*}
\begin{split}
\left[\hat{b}_{1,j}, \dots,  \hat{b}_{N,j}\right]^T & = \mathbf{U}_j^{-1} \left[\hat{d}_{1,j}, \dots,  \hat{d}_{N,j}\right]^T  \\ & \stackrel{\textrm{a.s.}}{\rightarrow}  \mathbf{U}^{-1} \mathbf{U} \left[b_1, \dots, b_N \right]^T  =  \left[b_1, \dots, b_N \right]^T.
\end{split}
\end{equation*}
\end{proof}

\subsection{Simulation Results}
We first consider the same third order system as in Section \ref{dumb_quantizer_sim_sec}, where $b_1=0.2$, $b_2=-0.2$, $b_3=0.6$, and the inputs and noises are Gaussian with $\mu=1$, $\sigma_u^2=1$, $\sigma_w^2=1$. We use the identification schemes in Algorithms 3 and 4. 
In the schemes we use the sequences $\alpha_j = \frac{1}{j}$, and the threshold $c_u=1$. Figs. \ref{smart_quantizer_alg3_plot} and \ref{smart_quantizer_alg4_plot_gaussian} shows the estimates $\hat{b}_1, \hat{b}_2, \hat{b}_3$ from Algorithms 3 and 4 respectively. 
\begin{figure}[t!]
\centering 
\includegraphics[scale=0.5]{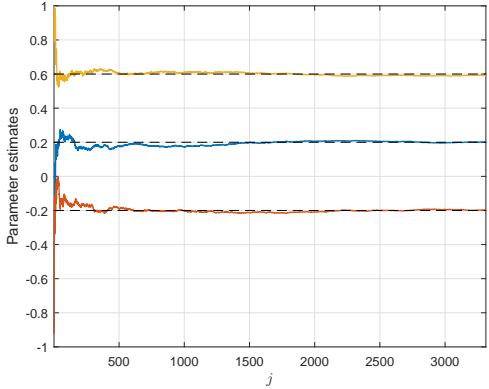} 
\caption{Parameter Estimates: Algorithm 3}
\label{smart_quantizer_alg3_plot}
\end{figure} 

\begin{figure}[t!]
\centering 
\includegraphics[scale=0.5]{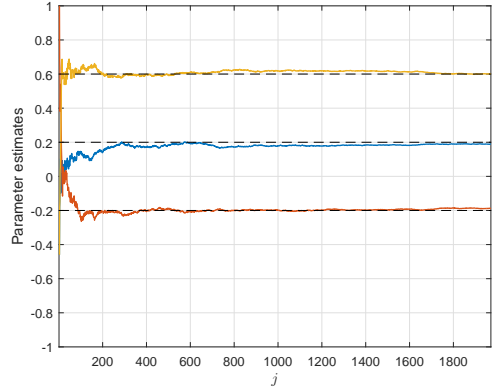} 
\caption{Parameter Estimates: Algorithm 4 with Gaussian inputs and noise}
\label{smart_quantizer_alg4_plot_gaussian}
\end{figure}

Next, we change $w_t$ to be uniformly distributed between $-\sqrt{3}$ and $\sqrt{3}$, and $u_t$ to be uniformly distributed between $0$ and $2\sqrt{3}$ (so that the variances are equal to 1). Fig. \ref{smart_quantizer_alg4_plot_uniform}  shows the estimates $\hat{b}_1, \hat{b}_2, \hat{b}_3$ from Algorithm 4. 
\begin{figure}[t!]
\centering 
\includegraphics[scale=0.5]{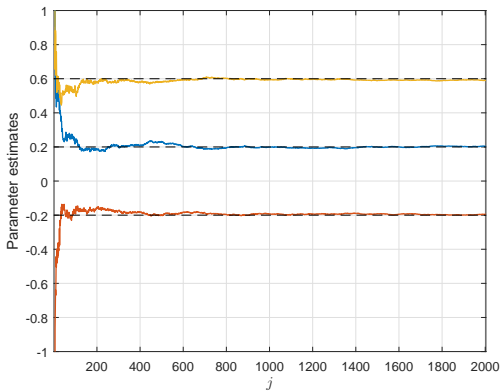} 
\caption{Parameter Estimates: Algorithm 4 with uniformly distributed inputs and noise}
\label{smart_quantizer_alg4_plot_uniform}
\end{figure}

Approximations of the variances $t^{1/2} (\hat{b}_{n,j}-b_n), n=1,2,3$ using Monte Carlo approximations over 10000 simulation runs are plotted for the Gaussian distributed inputs and noise case with Algorithms 3 and 4  in Figs. \ref{smart_quantizer_alg3_convergence_plot} and \ref{smart_quantizer_alg4_convergence_plot_gaussian} respectively, and for uniformly distributed inputs and noise in Fig. \ref{smart_quantizer_alg4_convergence_plot_uniform}. Comparing Figs. \ref{smart_quantizer_alg3_convergence_plot} and \ref{smart_quantizer_alg4_convergence_plot_gaussian} with Figs. \ref{dumb_quantizer_alg1_convergence_plot} and \ref{dumb_quantizer_alg2_convergence_plot}, we see that the normalized variances are much smaller, and hence convergence of the algorithms is better, when the quantizers have some computational and storage capabilities. We do emphasize however that the algorithms are based on different principles, so it is not a straightforward comparison. 
\begin{figure}[t!]
\centering 
\includegraphics[scale=0.5]{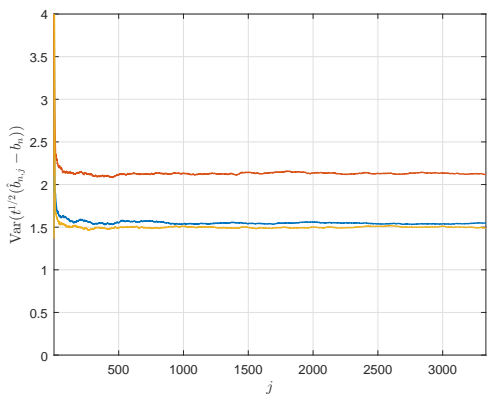} 
\caption{Convergence Behaviour: Algorithm 3 with Gaussian inputs and noise}
\label{smart_quantizer_alg3_convergence_plot}
\end{figure} 

\begin{figure}[t!]
\centering 
\includegraphics[scale=0.5]{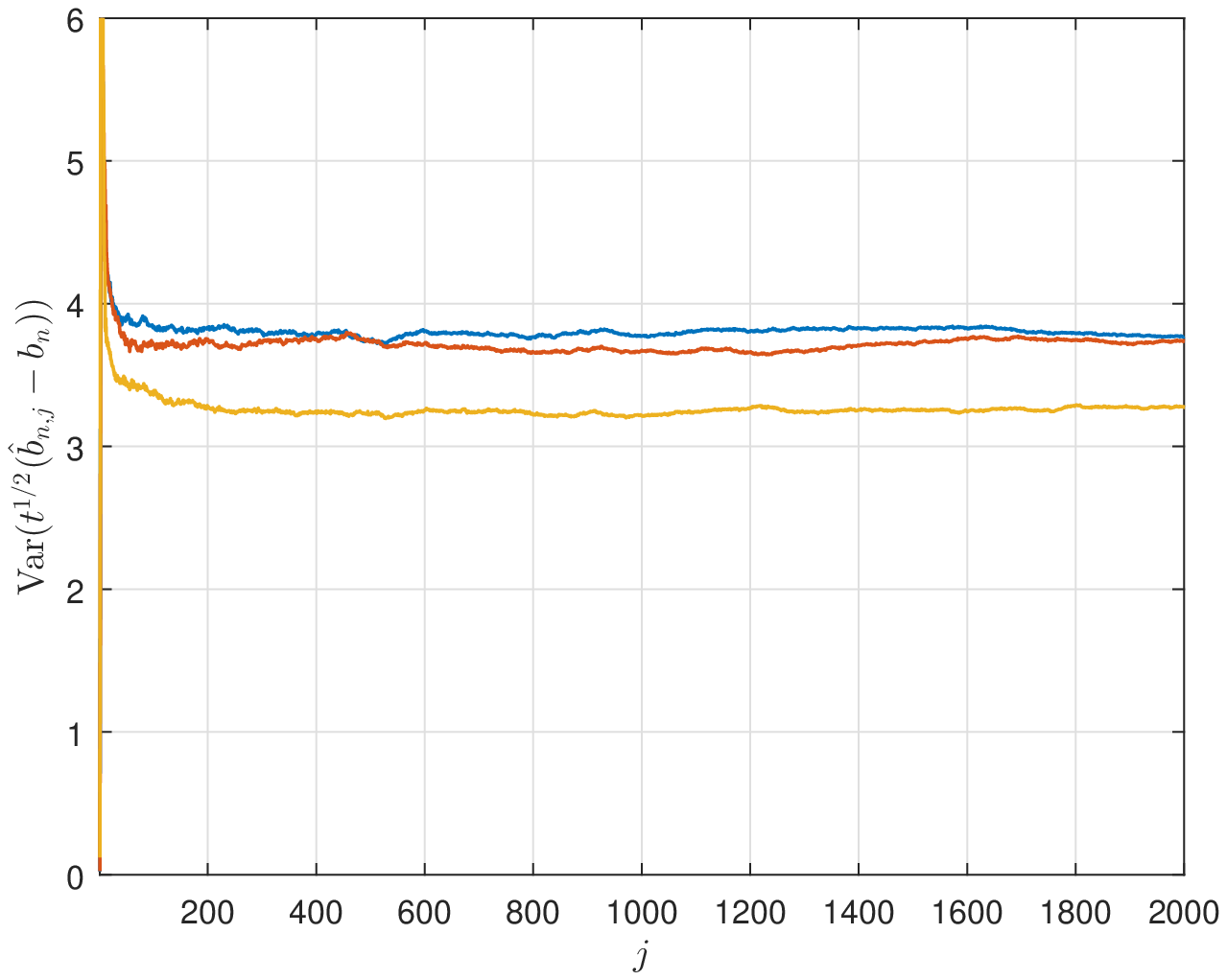} 
\caption{Convergence Behaviour: Algorithm 4 with Gaussian inputs and noise}
\label{smart_quantizer_alg4_convergence_plot_gaussian}
\end{figure} 

\begin{figure}[t!]
\centering 
\includegraphics[scale=0.5]{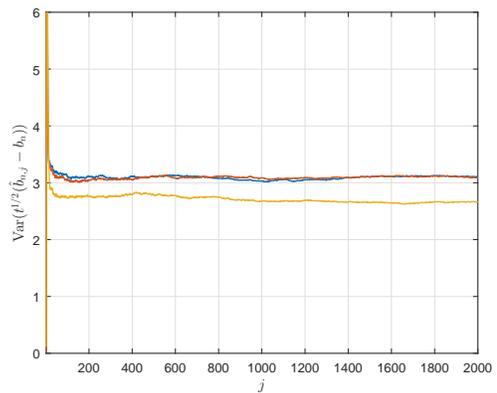} 
\caption{Convergence Behaviour: Algorithm 4 with uniformly distributed inputs and noise}
\label{smart_quantizer_alg4_convergence_plot_uniform}
\end{figure}

\begin{table*}[t!]
\caption{Summary of identification algorithms}
\centering
\begin{tabular}{|c|c|c|c|c|} \hline
 Algorithm & Computational Capability & Input Signal & Output Noise &  Input Parameter   \\
 			& of Quantizer & &  & Knowledge  \\ \hline \hline
  Algorithm 1 & None & i.i.d. Gaussian & i.i.d. zero-mean Gaussian   & $(\mu,\sigma_u^2)$  \\ \hline
  Algorithm 2 & None & i.i.d. Gaussian & i.i.d. zero-mean Gaussian  & None \\ \hline
Algorithm 3 & At output quantizer & i.i.d. & i.i.d. zero-mean   & $\mathbb{E}[u]$ \& $\mathbb{E}[u|u>c_u]$  \\ \hline
Algorithm 4 & At input \& output quantizers & i.i.d. & i.i.d. zero-mean    & None  \\ \hline
\end{tabular}
\label{summary_table}
\end{table*}

\section{Conclusion}
This paper has considered the identification of FIR systems with binary input and output observations. For the case where the quantizer thresholds can be adapted but the quantizers have no computational capabilities, we proposed identification schemes which are strongly consistent for Gaussian distributed inputs and noises. For the case of smart quantizers which have some computational and storage capabilities, strongly consistent identification schemes are proposed which can handle arbitrary input and noise distributions.  A summary of the main features and assumptions required for the different algorithms is provided in Table \ref{summary_table}.  Numerical simulations have illustrated the performance of the algorithms.  Rigorous analyses of the convergence rates of the algorithms is currently under investigation.  

\begin{appendix}

\subsection{Proof of Lemma \ref{F_increasing_lemma}}
\label{F_increasing_lemma_proof}
First note that 
\begin{align*}
& \frac{\partial}{\partial b}\Bigg[1 - \Phi \Bigg( \frac{- b(u - \mu)}{\sqrt{\widehat{\textrm{V}y} - b^2 \sigma_u^2}}\Bigg) \Bigg] =  \frac{1}{\sqrt{2\pi}}\exp\left( - \frac{b^2 (u-\mu)^2}{2(\widehat{\textrm{V}y} - b^2 \sigma_u^2)} \right) \\ & \quad \times \left[ \frac{b^2 \sigma_u^2 (u-\mu)}{ (\widehat{\textrm{V}y}-b^2 \sigma_u^2)^{3/2}} + \frac{u-\mu}{ (\widehat{\textrm{V}y}-b^2 \sigma_u^2)^{1/2} } \right]
 \end{align*}
is continuous for $b \in \Big(  - \sqrt{\widehat{\textrm{V}y}/\sigma_u^2},  \sqrt{\widehat{\textrm{V}y}/\sigma_u^2}\Big)$,  since quotients and compositions of continuous functions are continuous. For fixed $\widehat{\textrm{V}y}$, regard $F(b, \widehat{\textrm{V}y})$ as a function of $b$. Then by the Leibniz rule, we have that 
\begin{equation*}
\begin{split}
\frac{dF}{db} 
& = \int_{\mu}^\infty \frac{\partial}{\partial b}\Bigg[1 - \Phi \Bigg( \frac{- b(u - \mu)}{\sqrt{\widehat{\textrm{V}y} - b^2 \sigma_u^2}}\Bigg) \Bigg] p(u) du \\
& = \int_{\mu}^\infty \frac{1}{\sqrt{2\pi}}\exp\left( - \frac{b^2 (u-\mu)^2}{2(\widehat{\textrm{V}y} - b^2 \sigma_u^2)} \right) \\ & \quad \times \left[ \frac{b^2 \sigma_u^2 (u-\mu)}{ (\widehat{\textrm{V}y}-b^2 \sigma_u^2)^{3/2}} + \frac{u-\mu}{ (\widehat{\textrm{V}y}-b^2 \sigma_u^2)^{1/2} } \right] p(u) du \\
& > 0, \quad \forall b \in \left(  - \sqrt{\widehat{\textrm{V}y}/\sigma_u^2},  \sqrt{\widehat{\textrm{V}y}/\sigma_u^2}\right),
\end{split}
\end{equation*}
since each term in the integrand is strictly positive for $b \in \Big(  - \sqrt{\widehat{\textrm{V}y}/\sigma_u^2},  \sqrt{\widehat{\textrm{V}y}/\sigma_u^2}\Big)$ and $u > \mu$. 

\subsection{Statement of Theorem  2.4.1(ii) of \cite{Chen_SA}}
\label{chen_thm_statement}
We provide here the statement of Theorem 2.4.1(ii) of \cite{Chen_SA}, adapted to the notation of this paper. A major part of the proof of Theorem \ref{dumb_quantizer_convergence_thm} is the verification of the conditions of this theorem. 
\begin{theorem}[Theorem 2.4.1(ii) of \cite{Chen_SA}]
Consider the procedure
\begin{equation*}
%\label{expanding_truncations}
x_{j+1} = \Pi_{M_{\varsigma(j)}} (x_j + \alpha_j (f(x_j) + \varepsilon_j))
\end{equation*}
where
\begin{equation*}
%\label{increasing_thresholds}
\varsigma(0) = 0, \quad \varsigma(j) \triangleq \sum_{i=1}^{j-1} \mathds{1}(||x_i + \alpha_i (f(x_i) + \varepsilon_i)|| > M_{\varsigma(i)}), 
\end{equation*}
and the truncation operation
\begin{equation*}
%\label{truncation_operation}
\Pi_M(x) \triangleq \left\{\begin{array}{ll} x, & ||x|| \leq M \\ x^*, &||x|| > M. \end{array}  \right.
\end{equation*}

Suppose $f(.): \mathbb{R}^l \rightarrow \mathbb{R}$ has a unique root $x^0$, and $f(.)$ is continuous at $x^0$. Further assume that conditions A2.2.1 and A2.2.2 below hold.\\
A2.2.1:   $\alpha_j > 0$, $\alpha_j \rightarrow 0$, and $\sum_{j=1}^\infty \alpha_j = \infty$.
 \\
 A2.2.2: There exists a continuously differentiable function $v(.): \mathbb{R}^l \rightarrow \mathbb{R}$ such that  
\begin{equation}
\label{lyapunov_condition}
\sup_{\delta \leq ||x-x^0|| \leq \Delta}  f^T(x)\nabla v(x) < 0
\end{equation}
for any $\Delta > \delta > 0$, and 
\begin{equation}
\label{lyapunov_condition2}
v(x^*) < \inf_{||x|| = c_0} v(x)
\end{equation}
 for some $c_0 > 0$ and $||x^*|| < c_0$.

Then $\{x_j\}$ converges to $x^0$ for those sample paths where $\varepsilon_j$ can be written as $\varepsilon_j = \varepsilon_j^{(1)} + \varepsilon_j^{(2)}$, with
$$\sum_{j=1}^\infty \alpha_j \varepsilon_j^{(1)} < \infty \textrm{ and } \varepsilon_j^{(2)} \rightarrow 0. $$
\end{theorem}

\subsection{Proof of Theorem \ref{dumb_quantizer_convergence_thm}}
\label{dumb_quantizer_convergence_thm_proof}
The idea of the proof is that we will first show that Algorithm 1 can be viewed as a multi-dimensional stochastic approximation algorithm with expanding truncations. We will then verify the conditions of Theorem 2.4.1(ii) of \cite{Chen_SA} (given in Appendix \ref{chen_thm_statement}) to conclude that  $\hat{b}_{n,j}  \stackrel{\textrm{a.s.}}{\rightarrow} b_n$ as $j \rightarrow \infty$ for $n=1,\dots,N$. 

Define $f(.): \mathbb{R}^{N+2} \rightarrow \mathbb{R}^{N+2}$ by
\begin{equation*}
f \left(\! \left[\begin{array}{c} c_y \\ \tilde{c}_y \\ \hat{b}_1 \\ \vdots \\ \hat{b}_N \end{array} \right] \!\right) \triangleq \left[\!\!\begin{array}{c} \mathbb{P} (y_t > c_y) - 0.5 \\  \mathbb{P} (y_t > \tilde{c}_y) - 0.1587 \\ F(b_1,(\tilde{c}_y - c_y)^2) - F(\hat{b}_1,(\tilde{c}_y - c_y)^2) \\ \vdots \\ F(b_N,(\tilde{c}_y - c_y)^2) - F(\hat{b}_N,(\tilde{c}_y - c_y)^2) \end{array}  \!\!\right].
\end{equation*}
where 
$$\mathbb{P}(y_t > c) = 1 - \Phi \bigg( \frac{c -(b_1 + \dots +  b_N) \mu}{\sqrt{(b_1^2 + \dots + b_N^2) \sigma_u^2 + \sigma_w^2}}\bigg)$$
is the stationary probability that $y_t > c$, and $F(.,.)$ is given by (\ref{F_fn}). 

Since $y_t$ is Gaussian, the unique solution to $\mathbb{P} (y_t > c_y) - 0.5 =0$ is clearly $c_y = \mathbb{E}[y]$, and the unique  solution to 
$\mathbb{P} (y_t > \tilde{c}_y) - 0.1587 =0 $ is clearly $\tilde{c}_y = \mathbb{E}[y] + \sqrt{\textrm{Var}[y]}$. For $c_y = \mathbb{E}[y]$ and $\tilde{c}_y = \mathbb{E}[y] + \sqrt{\textrm{Var}[y]}$, each of the equations
$$F(b_n,(\tilde{c}_y - c_y)^2) - F(\hat{b}_n,(\tilde{c}_y - c_y)^2) = 0 $$ has the unique solution $\hat{b}_n = b_n$ (i.e. the true value of the parameter) by Lemma \ref{F_increasing_lemma}.  Hence the equation 
\begin{equation}
\label{root_eqn}
f([c_y, \tilde{c}_y, \hat{b}_1, \dots, \hat{b}_N ]^T) = \mathbf{0}
%f \left(\left[\begin{array}{c} c_y \\ \tilde{c}_y \\ \hat{b}_1 \\ \vdots \\  \hat{b}_N \end{array} \right] \right)= \left[ \begin{array}{c} 0 \\ 0 \\ 0 \%\ \vdots \\ 0 \end{array}\right]
\end{equation}
 has the unique root 
\begin{equation}
\label{x0}
[c_y, \tilde{c}_y, \hat{b}_1, \dots, \hat{b}_N ]^T \!=\! [ \mathbb{E}[y],  \mathbb{E}[y] + \sqrt{\textrm{Var}[y]}, b_1, \dots, b_N]^T \! \triangleq \! x^0.
\end{equation}

Next, let us write the recursions in Algorithm 1 in the following form:
\begin{equation}
\label{alg1_SA_form}
\begin{split}
& \left[\!\!\!\begin{array}{c} c_{y,j+1} \\ \tilde{c}_{y,j+1} \\ \hat{b}_{1,j+1} \\ \vdots \\  \hat{b}_{N,j+1} \end{array}\!\!\! \right]   = \Pi_{M_{\varsigma(j)}} \left( \left[\!\!\!\begin{array}{c} c_{y,j} \\ \tilde{c}_{y,j} \\ \hat{b}_{1,j} \\ \vdots \\ \hat{b}_{N,j} \end{array} \!\!\!\right] 
 +  \alpha_j f \left(\! \left[\!\!\begin{array}{c} c_{y,j} \\ \tilde{c}_{y,j} \\ \hat{b}_{1,j} \\ \vdots \\ \hat{b}_{N,j} \end{array} \!\!\right] \!\right)   + \alpha_j \varepsilon_j \right)
\end{split}
\end{equation}
where
\begin{equation}
\label{dumb_quantizer_martingale_noise}
\begin{split}
&\varepsilon_j \triangleq \\ & \left[\!\!\begin{array}{c} \mathds{1} (y_{2j-1} > c_{y,j}) - \mathbb{P}(y_{t} > c_{y,j}) \\  \mathds{1} (y_{2j} > \tilde{c}_{y,j}) - \mathbb{P} (y_{t} > \tilde{c}_{y,j}) \\ \mathds{1}(u_{2j-2} > \mu) \mathds{1}(y_{2j-1} > c_{y,j}) -  F(b_1,(\tilde{c}_{y,j} - c_{y,j})^2)  \\ \vdots \\ \mathds{1}(u_{2j-1-N} \!>\! \mu) \mathds{1}(y_{2j-1} \!>\! c_{y,j}) -  F(b_N,(\tilde{c}_{y,j} \!-\! c_{y,j})^2)  \end{array}  \!\!\right],
\end{split}
\end{equation}
which is in the form of a multi-dimensional stochastic approximation algorithm, that tries to find the roots of the equation (\ref{root_eqn}),
with the ``noise'' term being $\varepsilon_j$.

We will now verify the conditions of Theorem 2.4.1(ii) of \cite{Chen_SA}, whose statement is also given in Appendix \ref{chen_thm_statement}.
We first need $f(.)$ to have a unique root $x^0$, with $f(.)$ continuous at $x^0$. Uniqueness of $x^0$ given by (\ref{x0}) has been shown at the beginning of the proof, and $f(.)$ is clearly continuous at $x^0$. 

Note that Condition A2.2.1 is true by assumption. 
We next  verify Condition A2.2.2. 
Choose $v(x) \triangleq ||x-x^0||^2$, i.e.
$$v\left([c_y, \tilde{c}_y,  \hat{b}_1, \dots, \hat{b}_N]^T \right) = \left\| [ c_y, \tilde{c}_y,  \hat{b}_1, \dots, \hat{b}_N ]^T - x_0 \right\|^2.$$
Then 
\begin{equation*}
\begin{split}
& f^T\left([c_y, \tilde{c}_y,  \hat{b}_1, \dots, \hat{b}_N]^T \right)  \nabla v_n\left([c_y, \tilde{c}_y,  \hat{b}_1, \dots, \hat{b}_N]^T \right) \\ & = 2 \left[\begin{array}{c} \mathbb{P} (y_t > c_y) - 0.5 \\  \mathbb{P} (y_t > \tilde{c}_y) - 0.1587 \\ F(b_1,(\tilde{c}_y - c_y)^2) - F(\hat{b}_1,(\tilde{c}_y - c_y)^2) \\ \vdots \\ F(b_N,(\tilde{c}_y - c_y)^2) - F(\hat{b}_N,(\tilde{c}_y - c_y)^2)\end{array}  \right]^T \\ & \qquad \qquad \times 
\left( [ c_y, \tilde{c}_y,  \hat{b}_1, \dots, \hat{b}_N ]^T - x_0 \right) \\
& < 0, \,\, \forall [ c_y, \tilde{c}_y,  \hat{b}_1, \dots, \hat{b}_N ]^T \neq x^0.
\end{split}
\end{equation*}
The inequality above holds since  $\mathbb{P} (y_t > c_y) - 0.5$ and  $ \mathbb{P} (y_t > \tilde{c}_y) - 0.1587$ are strictly decreasing in $c_y$ and $\tilde{c}_y$ respectively,
$$F(b_n,(\tilde{c}_y - c_y)^2) - F(\hat{b}_n,(\tilde{c}_y - c_y)^2)$$   for fixed $c_y$ and $\tilde{c}_y$ is  decreasing  in $\hat{b}_n, n=1,\dots,N$  by Lemma \ref{F_increasing_lemma}, and for $c_y = \mathbb{E}[y]$ and $ \tilde{c}_{y}  = \mathbb{E}[y] +\sqrt{\textrm{Var}[y]}$ only takes on the value 0 when $\hat{b}_n = b_n$. This verifies (\ref{lyapunov_condition}). Also, by the reverse triangle inequality, we have
$$v(x) = ||x-x^0||^2 \geq \Big| ||x|| - ||x^0|| \Big|^2$$
so that for some $c_0 > 2 ||x^0||$, one has 
$$ \inf_{||x|| = c_0} v(x) > (2  ||x^0|| -  ||x^0||)^2 = ||x^0||^2 = v(\mathbf{0}). $$ 
By our choice of $x^* = \mathbf{0}$, this verifies (\ref{lyapunov_condition2}) and hence condition A2.2.2. 

Finally, we want to show that $\varepsilon_j$ can be written as $\varepsilon_j = \varepsilon_j^{(1)} + \varepsilon_j^{(2)}$, with
$$\sum_{j=1}^\infty \alpha_j \varepsilon_j^{(1)} < \infty \textrm{ a.s. and } \varepsilon_j^{(2)} \rightarrow 0 \textrm{ a.s.}, $$
which will then imply the a.s. convergence of $[c_{y,j}, \tilde{c}_{y,j}, \hat{b}_{1,j},\dots,\hat{b}_{N,j}]^T$ to $x^0$. 
Rewrite (\ref{dumb_quantizer_martingale_noise}) as
\begin{equation*}
\begin{split}
&\varepsilon_j  = \\ &\left[\begin{array}{l} \mathds{1} (y_{2j-1} > c_{y,j}) - \mathbb{P}(y_{2j-1} > c_{y,j}) \\  \mathds{1} (y_{2j} > \tilde{c}_{y,j}) - \mathbb{P} (y_{2j} > \tilde{c}_{y,j}) \\ \mathds{1}(u_{2j-2} > \mu) \mathds{1}(y_{2j-1} > c_{y,j}) \\  \qquad -  \mathbb{P}(u_{2j-2} > \mu, y_{2j-1} > c_{y,j}) \\ \qquad \qquad \vdots \\ \mathds{1}(u_{2j-1-N} > \mu) \mathds{1}(y_{2j-1} > c_{y,j}) \\  \qquad-  \mathbb{P}(u_{2j-1-N} > \mu, y_{2j-1} > c_{y,j})\end{array}  \right] +\\
&  
 \left[\!\! \begin{array}{l} \mathbb{P}(y_{2j-1} > c_{y,j}) - \mathbb{P}(y_{t} > c_{y,j})  \\   \mathbb{P} (y_{2j} > \tilde{c}_{y,j}) -  \mathbb{P} (y_{t} > \tilde{c}_{y,j}) \\ 
\mathbb{P}(u_{2j-2} > \mu, y_{2j-1} > c_{y,j}) - \mathbb{P}(u_{t-1} > \mu, y_{t} > c_{y,j}) 
\\ \qquad + \mathbb{P}(u_{t-1} > \mu, y_{t} > c_{y,j}) - F(b_1,(\tilde{c}_{y,j}-c_{y,j})^2) \\ \qquad\qquad \vdots \\ 
\mathbb{P}(u_{2j-1-N} > \mu, y_{2j-1} > c_{y,j}) - \mathbb{P}(u_{t-N} > \mu, y_{t} > c_{y,j}) \\  \qquad + \mathbb{P}(u_{t-N} > \mu, y_{t} > c_{y,j}) -  F(b_N,(\tilde{c}_{y,j}-c_{y,j})^2) \end{array} \!\! \right] \\ 
& \triangleq \left[\varepsilon_j^{(1,1)},  \varepsilon_j^{(2,1)}, \varepsilon_j^{(3,1)}, \dots, \varepsilon_j^{(N+2,1)}  \right]^T \\ & \quad + \left[ \varepsilon_j^{(1,2)},   \varepsilon_j^{(2,2)},  \varepsilon_j^{(3,2)}, \dots, \varepsilon_j^{(N+2,2)}  \right]^T.
\end{split}
\end{equation*}
We will prove that $\sum_{j=1}^\infty \alpha_j \varepsilon_j^{(i,1)} < \infty \textrm{ a.s.}$  and  $\varepsilon_j^{(i,2)}  \stackrel{\textrm{a.s.}}{\rightarrow} 0 $ for $i=1,2,\dots,N+2$. 

In order to show that $\sum_{j=1}^\infty \alpha_j \varepsilon_j^{(i,1)} < \infty \textrm{ a.s.}$, we will show that each $\{\varepsilon_j^{(i,1)} \}$ is a martingale difference sequence, which will then imply that $\sum_{j=1}^\infty \alpha_j \varepsilon_j^{(i,1)} < \infty \textrm{ a.s.}$, by e.g. Theorem B.6.1 of \cite{Chen_SA}. 
Define the  $\sigma$-algebras 
\begin{equation}
\label{sigma_algebra}
\begin{split}
&\mathcal{F}_j \triangleq  \sigma\big(\{\mathds{1}(y_{2i-1} > c_{y,i}), \mathds{1}(y_{2i} > \tilde{c}_{y,i}), \mathds{1}(u_{2i-2} > \mu), \\ &\qquad\qquad i=1,\dots,j  \}\big). 
\end{split}
\end{equation}
From the recursion for $c_{y,j}$, we note that $c_{y,j}$ is measurable with respect to $\mathcal{F}_j$ (and in fact is also measurable with respect to $\mathcal{F}_{j-1}$), and so $\varepsilon_j^{(1,1)}$ is measurable with respect to $\mathcal{F}_j$. 
We have 
\begin{equation*}
\begin{split}
\mathbb{E}[\varepsilon_j^{(1,1)}|\mathcal{F}_{j-1}] & = \mathbb{E} [ \mathds{1}(y_{2j-1} > c_{y,j}) - \mathbb{P}(y_{2j-1} > c_{y,j})|\mathcal{F}_{j-1}] \\ &=0.
\end{split}
\end{equation*}
Thus $\{\varepsilon_j^{(1,1)}\}$ is a martingale difference sequence.  Similar arguments can be used to show that $\{\varepsilon_j^{(i,1)}\}$ for  $i=2,\dots,N+2$ are martingale difference sequences, and therefore that $\sum_{j=1}^\infty \alpha_j \varepsilon_j^{(i,1)} < \infty \textrm{ a.s.}$ 

Let us now show that $\varepsilon_j^{(1,2)}  = \mathbb{P}(y_{2j-1} > c_{y,j}) - \mathbb{P}(y_{t} > c_{y,j})  \stackrel{\textrm{a.s.}}{\rightarrow} 0 $. 
First, we note that $y_{2j-1} $ and $c_{y,j-L}$ are independent for sufficiently large $L$, e.g. $L=N$, so that we can write
\begin{equation}
\label{y_c_independence}
 \mathbb{P}(y_{2j-1} > c_{y,j-L}) =  \mathbb{P}(y_{t} > c_{y,j-L}).
 \end{equation}
Next, we note that $c_{y,j+1}$ differs from $c_{y,j}$ by $\alpha_j/2$ (either above or below). Since by assumption $\alpha_j > 0, \forall j$, we can bound the difference between $c_{y,j}$ and $c_{y,j-L}$ as follows:
\begin{equation}
\label{cyN_bound}
 c_{y,j-L} - \sum_{i=j-L}^{j-1} \frac{\alpha_i}{2} \leq c_{y,j} \leq  c_{y,j-L} + \sum_{i=j-L}^{j-1} \frac{\alpha_i}{2},
 \end{equation}
 where $\sum_{i=j-L}^{j-1} \frac{\alpha_i}{2} $ is a deterministic quantity. 
Then we have
\begin{align*} & \mathbb{P}\bigg(y_{2j-1} > c_{y,j-L} + \sum_{i=j-L}^{j-1} \frac{\alpha_i}{2}\bigg) \leq \mathbb{P}(y_{2j-1} > c_{y,j}) \\ & \qquad \leq \mathbb{P}\bigg(y_{2j-1} > c_{y,j-L} - \sum_{i=j-L}^{j-1} \frac{\alpha_i}{2}\bigg), 
\end{align*}
or by (\ref{y_c_independence}) that 
\begin{align*} &  \mathbb{P}\bigg(y_{t} > c_{y,j-L} + \sum_{i=j-L}^{j-1} \frac{\alpha_i}{2}\bigg) \leq \mathbb{P}(y_{2j-1} > c_{y,j}) \\& \qquad\leq \mathbb{P}\bigg(y_{t} > c_{y,j-L} - \sum_{i=j-L}^{j-1} \frac{\alpha_i}{2}\bigg). 
\end{align*}
From (\ref{cyN_bound}) we also have 
\begin{align*}  & \mathbb{P}\bigg(y_{t} > c_{y,j-L} + \sum_{i=j-L}^{j-1} \frac{\alpha_i}{2}\bigg) \leq \mathbb{P}(y_{t} > c_{y,j}) \\ & \qquad \leq \mathbb{P}\bigg(y_{t} > c_{y,j-L} - \sum_{i=j-L}^{j-1} \frac{\alpha_i}{2}\bigg). 
\end{align*}
Thus
\begin{equation*}
\begin{split}
& \mathbb{P}\bigg(y_{t} > c_{y,j-L} +\! \!\!\sum_{i=j-L}^{j-1} \!\! \frac{\alpha_i}{2}\bigg) - \mathbb{P}\bigg(y_{t} > c_{y,j-L} - \!\!\!\sum_{i=j-L}^{j-1} \!\! \frac{\alpha_i}{2}\bigg) \\ & \leq \mathbb{P}(y_{2j-1} > c_{y,j}) -  \mathbb{P}(y_{t} > c_{y,j}) \\ &  \leq \mathbb{P}\bigg(y_{t} > c_{y,j-L} -\!\!\! \sum_{i=j-L}^{j-1} \!\!\frac{\alpha_i}{2}\bigg) - \mathbb{P}\bigg(y_{t} > c_{y,j-L} +\!\!\! \sum_{i=j-L}^{j-1} \!\!\frac{\alpha_i}{2}\bigg).
\end{split}
\end{equation*}
Since $\alpha_j \rightarrow 0$ as $j\rightarrow \infty$, we  also have $\sum_{i=j-L}^{j-1} \frac{\alpha_i}{2}  \rightarrow 0$ as $j\rightarrow \infty$. As $y_t$ is Gaussian, we then have 
$$\mathbb{P}\bigg(y_{t} > c_{y,j-L} + \!\!\sum_{i=j-L}^{j-1} \!\!\frac{\alpha_i}{2}\bigg) - \mathbb{P}\bigg(y_{t} > c_{y,j-L} - \!\!\sum_{i=j-L}^{j-1} \!\!\frac{\alpha_i}{2}\bigg)  \stackrel{\textrm{a.s.}}{\rightarrow} 0, $$ 
$$\mathbb{P}\bigg(y_{t} > c_{y,j-L} - \!\!\sum_{i=j-L}^{j-1} \!\!\frac{\alpha_i}{2}\bigg) - \mathbb{P}\bigg(y_{t} > c_{y,j-L} + \!\! \sum_{i=j-L}^{j-1} \!\!\frac{\alpha_i}{2}\bigg)  \stackrel{\textrm{a.s.}}{\rightarrow} 0, $$ 
and hence 
$$\mathbb{P}(y_{2j-1} > c_{y,j}) -  \mathbb{P}(y_{t} > c_{y,j})  \stackrel{\textrm{a.s.}}{\rightarrow} 0. $$ 
 By applying Theorem 2.4.1 of \cite{Chen_SA} to the recursions for $c_{y,j}$, we can then conclude that $c_{y,j}  \stackrel{\textrm{a.s.}}{\rightarrow} \mathbb{E}[y]$. 
 A  similar argument can be used to show  $\varepsilon_j^{(2,2)}   \stackrel{\textrm{a.s.}}{\rightarrow} 0 $, and hence that $\tilde{c}_{y,j}  \stackrel{\textrm{a.s.}}{\rightarrow} \mathbb{E}[y] +\sqrt{\textrm{Var}[y]}$. 
Moreover, it also follows that for $n=1,\dots,N$,
$$\mathbb{P}(u_{2j-1-n} > \mu, y_{2j-1} > c_{y,j}) -  \mathbb{P}(u_{t-n} > \mu, y_{t} > c_{y,j}) \stackrel{\textrm{a.s.}}{\rightarrow} 0 $$ by a similar argument. Next, the a.s. convergence to $0$ of
\begin{equation*}
\begin{split}
 & \mathbb{P}(u_{t-n} > \mu, y_{t} > c_{y,j}) - F(b_n,(\tilde{c}_{y,j}-c_{y,j})^2) \\ & = \int_{\mu}^\infty \bigg[1 - \Phi  \bigg( \frac{c_{y,j} - b_n (u - \mu) - \mathbb{E}[y]}{\sqrt{\textrm{Var}[y] - b_{n}^2 \sigma_u^2}}\bigg) \bigg] p(u) du \\ & \quad - F(b_n,(\tilde{c}_{y,j}-c_{y,j})^2)
\end{split}
\end{equation*}
 follows from the almost sure convergence of $c_{y,j}$ and $\tilde{c}_{y,j}$, and continuity.  Hence for $n=1,\dots,N$,
\begin{equation*}
\begin{split}
&\varepsilon_j^{(n+2,2)}  = \\ & \mathbb{P}(u_{2j-1-n} > \mu, y_{2j-1} > c_{y,j}) -  \mathbb{P}(u_{t-n} > \mu, y_{t} > c_{y,j}) \\ & \quad +  \mathbb{P}(u_{t-n} > \mu, y_{t} > c_{y,j})  - F(b_n,(\tilde{c}_{y,j}-c_{y,j})^2)\\  &  \stackrel{\textrm{a.s.}}{\rightarrow} 0.
\end{split}
\end{equation*}
By  Theorem 2.4.1 of \cite{Chen_SA} again, we then conclude the almost sure convergence of $[c_{y,j}, \tilde{c}_{y,j}, \hat{b}_{1,j},\dots,\hat{b}_{N,j}]^T$ to $[\mathbb{E}[y], \mathbb{E}[y] + \sqrt{\textrm{Var}[y]}, b_1, \dots,b_N]^T$ as $j\rightarrow \infty$, and in particular the almost sure convergence of $\hat{b}_{n,j}$ to the true value $b_n$, for $n=1,\dots,N$.

\subsection{Proof of Theorem \ref{dumb_quantizer_convergence_thm2}}
\label{dumb_quantizer_convergence_thm2_proof}
Note that in Algorithm 2, the update for $\hat{b}_{n,j+1}$ involves ``delayed'' information $c_{u,\bar{j}(n,j)}$ rather than $c_{u,j}$. We will first consider the convergence for a non-delayed version of Algorithm 2, and then  describe how delays can be handled. 

We first look at the recursions (\ref{alg2_recursions}), but with $c_{u,\bar{j}(n,j)}$ replaced by $c_{u,j}$ for $n=1,\dots,N$. 
Define $f(.): \mathbb{R}^{N+4} \rightarrow \mathbb{R}^{N+4}$ by
\begin{equation*}
f \!\left(\! \left[\!\!\begin{array}{c}  c_y \\ \tilde{c}_y  \\ c_u \\ \tilde{c}_u \\ \hat{b}_1 \\ \vdots \\ \hat{b}_N \end{array} \!\! \right]\! \right) \triangleq \left[\!\!\begin{array}{c} \mathbb{P} (y_t > c_y) - 0.5 \\ \mathbb{P} (y_t > \tilde{c}_y) - 0.1587  \\  \mathbb{P} (u_t > c_u) -0.5 \\  \mathbb{P} (u_t > \tilde{c}_u) - 0.1587 \\  h(c_y, \tilde{c}_y, c_u, \tilde{c}_u,b_1) - h(c_y, \tilde{c}_y, c_u, \tilde{c}_u,\hat{b}_1) \\ \vdots 
\\ h(c_y, \tilde{c}_y, c_u, \tilde{c}_u,b_N) \!-\! h(c_y, \tilde{c}_y, c_u, \tilde{c}_u,\hat{b}_N)\end{array} \!\! \right].
\end{equation*}
where $h(\cdot,\cdot,\cdot,\cdot,\cdot)$ is given by (\ref{h_defn}). By similar arguments as in the proof of Theorem \ref{dumb_quantizer_convergence_thm}, we can show that the equation
\begin{equation*}
%\label{root_eqn}
f([c_y, \tilde{c}_y, c_u, \tilde{c}_u,\hat{b}_1, \dots, \hat{b}_N ]^T) = \mathbf{0}
\end{equation*}
 has the unique root 
\begin{align*}
&[c_y, \tilde{c}_y, c_u, \tilde{c}_u, \hat{b}_1, \dots, \hat{b}_N ]^T \\&= [ \mathbb{E}[y],  \mathbb{E}[y] \!+\! \sqrt{\textrm{Var}[y]},  \mathbb{E}[u],  \mathbb{E}[u] \!+\! \sqrt{\textrm{Var}[u]}, b_1, \dots, b_N]^T.
\end{align*}

We then write the recursions  in the following form:
\begin{equation*}
%\label{alg2_non_delayed}
\begin{split}
& \left[ \!\!\!\begin{array}{c} c_{y,j+1} \\ \tilde{c}_{y,j+1} \\ c_{u,j+1} \\ \tilde{c}_{u,j+1} \\ \hat{b}_{1,j+1} \\ \vdots \\ \hat{b}_{N,j+1}  \end{array} \!\!\! \right]  = \Pi_{M_{\varsigma(j)}}\left( \left[\!\!\! \begin{array}{c} c_{y,j} \\ \tilde{c}_{y,j} \\  c_{u,j} \\ \tilde{c}_{u,j} \\ \hat{b}_{1,j} \\ \vdots \\ \hat{b}_{N,j} \end{array} \!\!\! \right] + \alpha_j  f \!\left(\! \left[\!\!\begin{array}{c}  c_{y,j} \\ \tilde{c}_{y,j}  \\ c_{u,j} \\ \tilde{c}_{u,j} \\ \hat{b}_{1,j} \\ \vdots \\ \hat{b}_{N,j} \end{array} \!\! \right]\! \right) + \alpha_j \varepsilon_j \right)
\end{split}
\end{equation*}
where
\begin{align*}
&\varepsilon_j \!=\! \left[\!\! \begin{array}{l} \mathds{1} (y_{2j-1} > c_{y,j}) -\mathbb{P} (y_t > c_{y,j}) \\ \mathds{1} (y_{2j} > \tilde{c}_{y,j}) - \mathbb{P} (y_t > \tilde{c}_{y,j})  \\ \mathds{1} (u_{2(j-1)+[j]_2} > c_{u,j}) - \mathbb{P} (u_t > c_{u,j}) \\  \mathds{1} (u_{2j+[j]_2} > \tilde{c}_{u,j}) - \mathbb{P} (u_t > \tilde{c}_{u,j}) \\  g(1,j) \big[\mathds{1} (u_{2j-2} > c_{u,j}) \mathds{1} (y_{2j-1} > c_{y,j}) \\ \qquad - h(c_{y,j}, \tilde{c}_{y,j},c_{u,j}, \tilde{c}_{u,j}, \hat{b}_{1,j}) \big] \\ \qquad -  h(c_{y,j}, \tilde{c}_{y,j}, c_{u,j}, \tilde{c}_{u,j},b_1) \\ \qquad + h(c_{y,j}, \tilde{c}_{y,j}, c_{u,j}, \tilde{c}_{u,j},\hat{b}_{1,j}) \\ \qquad \qquad \vdots 
\\ g(N,j) \big[\mathds{1} (u_{2j-1-N} \!>\! c_{u,j}) \mathds{1} (y_{2j-1} \!>\! c_{y,j}) 
\\ \qquad - h(c_{y,j}, \tilde{c}_{y,j},c_{u,j}, \tilde{c}_{u,j}, \hat{b}_{N,j}) \big] \\ \qquad - h(c_{y,j}, \tilde{c}_{y,j}, c_{u,j}, \tilde{c}_{u,j},b_N) \\ \qquad + h(c_{y,j}, \tilde{c}_{y,j}, c_{u,j}, \tilde{c}_{u,j},\hat{b}_{N,j}) \end{array} \!\! \right] 
\!\triangleq \!
 \left[\!\!\!\begin{array}{c} \varepsilon_j^1 \\ \varepsilon_j^2 \\ \varepsilon_j^3 \\ \varepsilon_j^4 \\ \varepsilon_j^5 \\ \vdots \\ \varepsilon_j^{N+4}  \end{array}\! \!\!\right]
\end{align*}
The first four components of $\varepsilon_j$ can be rewritten as
\begin{equation*}
\begin{split}
& \left[\!\!\begin{array}{c} \varepsilon_j^1 \\ \varepsilon_j^2 \\ \varepsilon_j^3 \\ \varepsilon_j^4 \end{array}\!\! \right]   = \left[\!\!\begin{array}{c} \mathds{1} (y_{2j-1} > c_{y,j}) -\mathbb{P} (y_{2j-1} > c_{y,j}) \\ \mathds{1} (y_{2j} > \tilde{c}_{y,j}) - \mathbb{P} (y_{2j} > \tilde{c}_{y,j})  \\ \mathds{1} (u_{2(j-1)+[j]_2} > c_{u,j}) - \mathbb{P} (u_{2(j-1)+[j]_2} > c_{u,j}) \\  \mathds{1} (u_{2j+[j]_2} > \tilde{c}_{u,j}) - \mathbb{P} (u_{2j+[j]_2} > \tilde{c}_{u,j}) \end{array} \!\! \right] 
\\ & \qquad \qquad +  \left[\!\!\begin{array}{c} \mathbb{P} (y_{2j-1} > c_{y,j}) - \mathbb{P} (y_t > c_{y,j})  
\\ \mathbb{P} (y_{2j} > \tilde{c}_{y,j}) -  \mathbb{P} (y_t > \tilde{c}_{y,j}) 
\\  \mathbb{P} (u_{2(j-1)+[j]_2} > c_{u,j})  -  \mathbb{P} (u_t > c_{u,j}) 
\\  \mathbb{P} (u_{2j+[j]_2} > \tilde{c}_{u,j}) -  \mathbb{P} (u_t > \tilde{c}_{u,j}) 
\end{array} \!\!  \right]
\\ & \triangleq \left[\varepsilon_j^{(1,1)},  \varepsilon_j^{(2,1)}, \varepsilon_j^{(3,1)}, \varepsilon_j^{(4,1)} \right]^T \!\!+\! \left[ \varepsilon_j^{(1,2)},   \varepsilon_j^{(2,2)},  \varepsilon_j^{(3,2)}, \varepsilon_j^{(4,2)} \right]^T\!\!.
\end{split}
\end{equation*}
By similar arguments as in the proof of Theorem \ref{dumb_quantizer_convergence_thm}, we can show that  $\sum_{j=1}^\infty \alpha_j \varepsilon_j^{(i,1)} < \infty \textrm{ a.s.}$  and  $\varepsilon_j^{(i,2)}  \stackrel{\textrm{a.s.}}{\rightarrow} 0 $ for $i=1,2,3,4$, and hence the almost sure convergence of $[c_{y,j}, \tilde{c}_{y,j}, c_{u,j}, \tilde{c}_{u,j}]^T$ to $[\mathbb{E}[y], \mathbb{E}[y] + \sqrt{\textrm{Var}[y]},\mathbb{E}[u], \mathbb{E}[u] + \sqrt{\textrm{Var}[u]}]^T$ as $j\rightarrow \infty$.

For the convergence of $\hat{b}_{n,j}, n=1,\dots,N$, note that if $g(n,j) = 1$, then $g(n,j+1)=0$, $g(n,j+2)=1$, $g(n,j+3)=0$ etc., so that $\hat{b}_{n,j}$ updates at every second $j$. When $g(n,j) = 1$, we have 
\begin{align*}
&\varepsilon_j^{n+4} \\& = \mathds{1} (u_{2j\!-\!1\!-\!n} \!>\! c_{u,j}) \mathds{1} (y_{2j\!-\!1}\! >\! c_{y,j}) 
 \!-\! h(c_{y,j}, \tilde{c}_{y,j}, c_{u,j}, \tilde{c}_{u,j},b_n)  \\
& = \big[\mathds{1} (u_{2j-1-n} > c_{u,j}) \mathds{1} (y_{2j-1} > c_{y,j}) \\ & \quad - \mathbb{P} (u_{2j-1-n} > c_{u,j}, y_{2j-1} > c_{y,j}) \big] \\
& \, + \big[\mathbb{P} (u_{2j-1-n} \!>\! c_{u,j}, y_{2j-1} \!>\! c_{y,j})  \!-\! \mathbb{P} (u_{t-n} \!>\! c_{u,j}, y_{t} \!>\! c_{y,j}) \\
& \qquad  + \mathbb{P} (u_{t-n} \!>\! c_{u,j}, y_{t} \!>\! c_{y,j})  \!-\! h(c_{y,j}, \tilde{c}_{y,j}, c_{u,j}, \tilde{c}_{u,j},b_n) \big]
\\ & \triangleq \varepsilon_j^{(n+4,1)} + \varepsilon_j^{(n+4,2)}
\end{align*}
For a given $n$, let $j_0(n)$ be the smallest positive integer such that $g(n,j_0(n)) = 1$. Using similar arguments as in the proof of Theorem \ref{dumb_quantizer_convergence_thm}, we can show that  $\sum_{j'=0}^\infty \alpha_{j_0(n)+2j'}^{\phantom{(n+4,1)}} \varepsilon_{j_0(n)+2j'}^{(n+4,1)} < \infty \textrm{ a.s.}$  and  $\varepsilon_{j_0(n)+2j'}^{(n+4,2)}  \stackrel{\textrm{a.s.}}{\rightarrow} 0 $ as $j' \rightarrow \infty$, for $n=1,2,\dots,N$. This then implies that $\hat{b}_{n,j_0(n)+2j'} \rightarrow b_n$ as $j' \rightarrow \infty$. As $\hat{b}_{n,j_0(n)+2j'+1} = \hat{b}_{n,j_0(n)+2j'}$, we also have $\hat{b}_{n,j_0(n)+2j'+1} \rightarrow b_n$ as $j' \rightarrow \infty$, and hence $\hat{b}_{n,j} \rightarrow b_n$ as $j \rightarrow \infty$. 

The above shows convergence for the recursion (\ref{alg2_recursions}), but with $c_{u,\bar{j}(n,j)}$ replaced by $c_{u,j}$ for $n=1,\dots,N$.  For the original updates (\ref{alg2_recursions}) in Algorithm 2 which uses the delayed information $c_{u,\bar{j}(n,j)}$, the situation can be considered as a case of the asynchronous stochastic approximation procedure of \cite[Sec. 5.6]{Chen_SA}.\footnote{The asynchronous stochastic approximation algorithm of \cite{Chen_SA} also allows for different step sizes $\alpha_j$ for each component, and different truncation times for different components.} Almost sure convergence of the procedure is shown by verifying conditions A5.6.1-A5.6.5 of \cite{Chen_SA}. Conditions  A5.6.1-A5.6.4 are similar to the conditions of Theorem 2.4.1 of \cite{Chen_SA}, and can be verified using similar arguments to the above, together with our assumption that the same $\alpha_j$ is used for all components. The additional condition is A5.6.5, which in our notation says that
\begin{equation*}
%\label{delayed_info_condition}
\lim_{j \rightarrow \infty} \sum_{i = \bar{j}(j,n)}^{j} \alpha_i \stackrel{\textrm{a.s.}}{=} 0, \quad n=1,\dots,N.
\end{equation*}
But this condition is true since $j - \bar{j}(j,n) = \lfloor \frac{n}{2} \rfloor$ is bounded and $\alpha_j \rightarrow 0$ as $j \rightarrow \infty$. 

\subsection{Proof of Theorem \ref{smart_quantizer_convergence_thm}}
\label{smart_quantizer_convergence_thm_proof}
As previously noted in (\ref{d_convergence}), we have that 
$d_{n,i}   \stackrel{\textrm{a.s.}}{\rightarrow} d_n$
for $n=1,\dots,N$. We will first show that 
\begin{equation}
\label{dhat_convergence2}
\begin{split}
\hat{d}_{n,j} \stackrel{\textrm{a.s.}}{\rightarrow}  d_n &= b_1  \mathbb{E}[u] + \dots + b_{n-1}  \mathbb{E}[u] + b_n \mathbb{E}[u | u > c] \\ & \quad+ b_{n+1}  \mathbb{E}[u] + \dots + b_N \mathbb{E}[u]
\end{split}
\end{equation}
for $n=1,\dots,N$, where $\hat{d}_{n,j}$ satisfies the recursion 
(\ref{sign_algorithm})
with $i=Nj$. 

Fix an arbitrary $n \in \{1,\dots,N\}$. 
Since $d_{n,i}   \stackrel{\textrm{a.s.}}{\rightarrow} d_n$, consider a sample path $\omega$ where $d_{n,i}   \rightarrow d_n$. 
We will show that one also has $\hat{d}_{n,j} \rightarrow d_n$ for this $\omega$. Let $\epsilon > 0$ be given. Since $d_{n,i}   \rightarrow d_n$, there exists an $i^* (\omega)$ dependent on $\omega$ such that 
\begin{equation}
\label{d_epsilon_property}
 | d_{n,i}   - d_n | < \frac{\epsilon}{2}, \quad \forall i \geq i^* (\omega).
\end{equation} 

Referring back to the recursion (\ref{sign_algorithm}), note that the iterate  $\hat{d}_{n,j+1} $ will either increase or decrease by $\alpha_j$ from the previous iterate $\hat{d}_{n,j} $, depending on whether $\hat{d}_{n,j} $ was below or above $d_{n,i}$ respectively. 
Let $j_0(i^*(\omega))$ be sufficiently large such that $j\geq j_0$ implies   $\alpha_j < \frac{\epsilon}{2}$ and $ i \geq i^* (\omega)$.  
We want to show that there exists a $j_1(i^*(\omega),\omega) \geq j_0$ such that $ | \hat{d}_{n,j_1}   - d_n | < \epsilon$. 
If $| \hat{d}_{n,j_0}   - d_n | < \epsilon$, then by setting $j_1= j_0$  we are done. If instead $| \hat{d}_{n,j_0}   - d_n | > \epsilon$, then such a $j_1$ exists since $ \sum_{j=j_0}^\infty \alpha_j = \infty$ (which follows from the assumption that $\sum_{j=0}^\infty \alpha_j = \infty$) and $\alpha_j \rightarrow 0$.

We next want to show that 
\begin{equation}
\label{dhat_epsilon_property}
| \hat{d}_{n,j_1}   - d_n | < \epsilon \Rightarrow | \hat{d}_{n,j_1+1}   - d_n | < \epsilon,
\end{equation}
 which by induction then implies  
 \begin{equation}
\label{dhat_epsilon_property2} 
| \hat{d}_{n,j}   - d_n | < \epsilon, \forall j \geq j_1.
\end{equation} 
There are two cases to consider: i) If  $| \hat{d}_{n,j_1}   - d_n | < \frac{\epsilon}{2}$, then 
$| \hat{d}_{n,j_1+1}   - d_n | < \epsilon$ since $\alpha_j < \frac{\epsilon}{2}$. ii) If $| \hat{d}_{n,j_1}   - d_n | > \frac{\epsilon}{2}$ and  $| \hat{d}_{n,j_1}   - d_n | < \epsilon$, then $\hat{d}_{n,j_1+1}$ will decrease by $\alpha_j $ if $\hat{d}_{n,j_1}   - d_n  > \frac{\epsilon}{2}$ (since $\hat{d}_{n,j_1} > d_{n,i}$ by (\ref{d_epsilon_property})), and increase by $\alpha_j $ if $\hat{d}_{n,j_1}   - d_n  < - \frac{\epsilon}{2}$. Either way, we have $| \hat{d}_{n,j_1+1}   - d_n | < \epsilon$.
Thus (\ref{dhat_epsilon_property2}) is satisfied, which means  that  $\hat{d}_{n,j} \rightarrow d_n$ for this $\omega$.
Therefore
$\mathbb{P} (\{\omega: \hat{d}_{n,j} \rightarrow d_n\} ) \geq \mathbb{P} (\{\omega: d_{n,i} \rightarrow d_n\} ) = 1$, and we have  $\hat{d}_{n,j} \stackrel{\textrm{a.s.}}{\rightarrow}  d_n$. Since $n$ was arbitrary, we thus have  $\hat{d}_{n,j} \stackrel{\textrm{a.s.}}{\rightarrow}  d_n$ for $n=1,\dots,N$. 
 
 To complete the proof, almost sure convergence of $[\hat{b}_{1,j},\dots,\hat{b}_{N,j}]$ to $[b_1,\dots,b_N]$ follows from (\ref{dhat_convergence2}) and continuity, since
\begin{equation*}
\begin{split}
\left[\hat{b}_{1,j}, \dots,  \hat{b}_{N,j}\right]^T & = \mathbf{U}^{-1} \left[\hat{d}_{1,j}, \dots,  \hat{d}_{N,j}\right]^T \\ & \stackrel{\textrm{a.s.}}{\rightarrow}  \mathbf{U}^{-1} \mathbf{U} \left[b_1, \dots, b_N \right]^T  =  \left[b_1, \dots, b_N \right]^T.
\end{split}
\end{equation*}

\end{appendix}

\bibliography{IEEEabrv,sysid}

\begin{thebibliography}{10}
\providecommand{\url}[1]{#1}
\csname url@samestyle\endcsname
\providecommand{\newblock}{\relax}
\providecommand{\bibinfo}[2]{#2}
\providecommand{\BIBentrySTDinterwordspacing}{\spaceskip=0pt\relax}
\providecommand{\BIBentryALTinterwordstretchfactor}{4}
\providecommand{\BIBentryALTinterwordspacing}{\spaceskip=\fontdimen2\font plus
\BIBentryALTinterwordstretchfactor\fontdimen3\font minus
  \fontdimen4\font\relax}
\providecommand{\BIBforeignlanguage}[2]{{%
\expandafter\ifx\csname l@#1\endcsname\relax
\typeout{** WARNING: IEEEtran.bst: No hyphenation pattern has been}%
\typeout{** loaded for the language `#1'. Using the pattern for}%
\typeout{** the default language instead.}%
\else
\language=\csname l@#1\endcsname
\fi
#2}}
\providecommand{\BIBdecl}{\relax}
\BIBdecl

\bibitem{Proakis}
J.~G. Proakis and M.~Salehi, \emph{Digital Communications}, 5th~ed.\hskip 1em
  plus 0.5em minus 0.4em\relax New York: McGraw-Hill, 2008.

\bibitem{Heinzelman_HICSS}
W.~R. Heinzelman, A.~Chandrakasan, and H.~Balakrishnan, ``Energy-efficient
  communication protocol for wireless microsensor networks,'' in \emph{Proc.
  {HICSS}}, Maui, HI, Jan. 2000.

\bibitem{Akyildiz}
I.~F. Akyildiz, W.~Su, Y.~Sankarasubramaniam, and E.~Cayirci, ``A survey on
  sensor networks,'' \emph{{IEEE} Commun. Mag.}, vol.~40, no.~8, pp. 102--114,
  Aug. 2002.

\bibitem{Wang_book}
L.~Y. Wang, G.~G. Yin, J.-F. Zhang, and Y.~Zhao, \emph{System Identification
  with Quantized Observations}.\hskip 1em plus 0.5em minus 0.4em\relax Boston:
  Birkhauser, 2010.

\bibitem{AgueroGoodwinYuz}
J.~C. Ag\"{u}ero, G.~C. Goodwin, and J.~I. Yuz, ``System identification using
  quantized data,'' in \emph{Proc. {IEEE} Conf. Decision and Control}, New
  Orleans, LA, 2007, pp. 4263--4268.

\bibitem{WeyerKoCampi}
E.~Weyer, S.~Ko, and M.~C. Campi, ``Finite sample properties of system
  identification with quantized output data,'' in \emph{Proc. {IEEE} Conf.
  Decision and Control}, Shanghai, China, Dec. 2009, pp. 1532--1537.

\bibitem{GodoyGoodwin}
B.~Godoy, G.~C. Goodwin, J.~C. Ag\"{u}ero, D.~Marelli, and T.~Wigren, ``On
  identification of {FIR} systems having quantized output data,''
  \emph{Automatica}, vol.~47, no.~9, pp. 1905--1915, Sep. 2011.

\bibitem{CasiniGarulliVicino_CDC}
M.~Casini, A.~Garulli, and A.~Vicino, ``Set-membership identification of {ARX}
  models with quantized measurements,'' in \emph{Proc. {IEEE} Conf. Decision
  and Control}, Orlando, FL, Dec. 2011, pp. 2806--2811.

\bibitem{You_adaptive}
K.~You, ``Recursive algorithms for parameter estimation with adaptive
  quantizer,'' \emph{Automatica}, vol.~52, pp. 192--201, Feb. 2015.

\bibitem{WangZhangYin}
L.~Y. Wang, J.-F. Zhang, and G.~G. Yin, ``System identification using binary
  sensors,'' \emph{{IEEE} Trans. Autom. Control}, vol.~48, no.~11, pp.
  1892--1907, Nov. 2003.

\bibitem{ColinetJuillard}
E.~Colinet and J.~Juillard, ``A weighted least-squares approach to parameter
  estimation problems based on binary measurements,'' \emph{{IEEE} Trans.
  Autom. Control}, vol.~55, no.~1, pp. 148--152, Jan. 2010.

\bibitem{CasiniGarulliVicino}
M.~Casini, A.~Garulli, and A.~Vicino, ``Input design in worse-case system
  identification using binary sensors,'' \emph{{IEEE} Trans. Autom. Control},
  vol.~56, no.~5, pp. 1186--1191, May 2011.

\bibitem{CsajiWeyer_CDC}
B.~C. Cs\'{a}ji and E.~Weyer, ``System identification with binary observations
  by stochastic approximation and active learning,'' in \emph{Proc. {IEEE}
  Conf. Decision and Control}, Orlando, FL, Dec. 2011, pp. 3634--3639.

\bibitem{CsajiWeyer_sysid}
------, ``Recursive estimation of {ARX} systems using binary sensors with
  adjustable threshold,'' in \emph{Proc. {IFAC} Symposium on System
  Identification}, Brussels, Belgium, Jul. 2012, pp. 1185--1190.

\bibitem{GoudjilPouliquen}
A.~Goudjil, M.~Pouliquen, E.~Pigeon, O.~Gehan, and M.~{M'Saad},
  ``Identification of systems using binary sensors via support vector
  machines,'' in \emph{Proc. {IEEE} Conf. Decision and Control}, Osaka, Japan,
  Dec. 2015, pp. 3385--3390.

\bibitem{PouliquenGoudjil}
M.~Pouliquen, A.~Goudjil, O.~Gehan, and E.~Pigeon, ``Continuous-time system
  identification using binary measurements,'' in \emph{Proc. {IEEE} Conf.
  Decision and Control}, Las Vegas, NV, Dec. 2016, pp. 3787--3792.

\bibitem{Ljung_book}
L.~Ljung, \emph{System Identification: Theory for the User}, 2nd~ed.\hskip 1em
  plus 0.5em minus 0.4em\relax New Jersey: Prentice Hall, 1999.

\bibitem{AbedMeraimQiuHua}
K.~Abed-Meraim, W.~Qiu, and Y.~Hua, ``Blind system identification,''
  \emph{Proc. {IEEE}}, vol.~85, no.~8, pp. 1310--1322, Aug. 1997.

\bibitem{YuZhangXie}
C.~Yu, C.~Zhang, and L.~Xie, ``Blind system identification using precise and
  quantized observations,'' \emph{Automatica}, vol.~49, pp. 2822--2830, 2013.

\bibitem{SuzukiSugie}
H.~Suzuki and T.~Sugie, ``System identification based on quantized {I/O} data
  corrupted with noises and its performance improvement,'' in \emph{Proc.
  {IEEE} Conf. Decision and Control}, San Diego, CA, Dec. 2006, pp. 3684--3689.

\bibitem{Ikenoue}
M.~Ikenoue, S.~Kanae, Z.-J. Yang, and K.~Wada, ``Identification of
  errors-in-variables models from quantized input-output measurements via
  bias-compensated instrumental variable type method,'' \emph{Int. J.
  Innovative Comput. Inform. Control}, vol.~6, no.~1, pp. 183--198, Jan. 2010.

\bibitem{GrayNeuhoff}
R.~M. Gray and D.~L. Neuhoff, ``Quantization,'' \emph{{IEEE} Trans. Inf.
  Theory}, vol.~44, no.~6, pp. 2325--2383, Oct. 1998.

\bibitem{CeronePigaRegruto}
V.~Cerone, D.~Piga, and D.~Regruto, ``Fixed-order {FIR} approximation of linear
  systems from quantized input and output data,'' \emph{Systems and Control
  Letters}, vol.~62, pp. 1136--1142, 2013.

\bibitem{Krishnamurthy_quantized}
V.~Krishnamurthy, ``Estimation of quantized linear errors-in-variables
  models,'' \emph{Automatica}, vol.~31, no.~10, pp. 1459--1464, 1995.

\bibitem{Kedem}
B.~Kedem, ``Estimation of the parameters in stationary autoregressive processes
  after hard limiting,'' \emph{J. Amer. Stat. Assoc.}, vol.~75, no. 369, pp.
  146--153, Mar. 1980.

\bibitem{GuoWang_FIR}
J.~Guo, L.~Y. Wang, G.~Yin, Y.~Zhao, and J.-F. Zhang, ``Asymptotically
  efficient identification of {FIR} systems with quantized observations and
  general quantized inputs,'' \emph{Automatica}, vol.~57, no.~1, pp. 113--122,
  2015.

\bibitem{YouWeyerNair}
K.~You, E.~Weyer, and G.~Nair, ``Identification of a gain system with binary
  input and output measurements,'' in \emph{Proc. {IEEE} Conf. Decision and
  Control}, Kyoto, Japan, Dec. 2015, pp. 2453--2458.

\bibitem{LianLuo}
Y.~Lian, Z.~Luo, E.~Weyer, and G.~N. Nair, ``Parameter estimation with binary
  observations of input and output signals,'' in \emph{Proc. {AUCC}},
  Newcastle, Australia, Nov. 2016, pp. 226--231.

\bibitem{Chen_Wiener_discontinuous}
H.-F. Chen, ``Recursive identification for {Wiener} model with discontinuous
  piece-wise linear function,'' \emph{{IEEE} Trans. Autom. Control}, vol.~51,
  no.~3, pp. 390--400, Mar. 2006.

\bibitem{MuChen_Wiener}
B.-Q. Mu and H.-F. Chen, ``Recursive identification of {MIMO} {Wiener}
  systems,'' \emph{{IEEE} Trans. Autom. Control}, vol.~58, no.~3, pp. 802--808,
  Mar. 2013.

\bibitem{ZhaoChen_Hammerstein}
W.~Zhao and H.-F. Chen, ``Adaptive tracking and recursive identification for
  {Hammerstein} systems,'' \emph{Automatica}, vol.~45, no.~12, pp. 2773--2783,
  2009.

\bibitem{MuChenWangYinZheng}
B.-Q. Mu, H.-F. Chen, L.~Y. Wang, G.~Yin, and W.~X. Zheng, ``Recursive
  identification of {Hammerstein} systems: Convergence rate and asymptotic
  normality,'' \emph{{IEEE} Trans. Autom. Control}, vol.~62, no.~7, pp.
  3277--3292, Jul. 2017.

\bibitem{ZhaoZhengBai}
W.~Zhao, W.~X. Zheng, and E.-W. Bai, ``A recursive local linear estimator for
  identification of nonlinear {ARX} systems: Asymptotical convergence and
  applications,'' \emph{{IEEE} Trans. Autom. Control}, vol.~58, no.~12, pp.
  3057--3069, Dec. 2013.

\bibitem{ZhaoChenTempoDabbene}
W.~Zhao, H.-F. Chen, R.~Tempo, and R.~Dabbene, ``Recursive nonparametric
  identification of nonlinear systems with adaptive binary sensors,''
  \emph{{IEEE} Trans. Autom. Control}, vol.~62, no.~8, pp. 3959--3971, Aug.
  2017.

\bibitem{LeongWeyerNair}
A.~S. Leong, E.~Weyer, and G.~N. Nair, ``On the identification of {FIR} systems
  with binary input and output observations,'' in \emph{Proc. {IEEE} Conf.
  Decision and Control}, Las Vegas, NV, Dec. 2016, pp. 2932--2937.

\bibitem{GrimmettStirzaker}
G.~R. Grimmett and D.~R. Stirzaker, \emph{Probability and Random Processes},
  3rd~ed.\hskip 1em plus 0.5em minus 0.4em\relax Oxford, UK: Oxford University
  Press, 2001.

\bibitem{KushnerYin}
H.~J. Kushner and G.~G. Yin, \emph{Stochastic Approximation and Recursive
  Algorithms and Applications}, 2nd~ed.\hskip 1em plus 0.5em minus 0.4em\relax
  New York: Springer, 2003.

\bibitem{Chen_SA}
H.-F. Chen, \emph{Stochastic Approximation and Its Applications}.\hskip 1em
  plus 0.5em minus 0.4em\relax Dordrecht, The Netherlands: Kluwer Academic
  Publishers, 2002.

\bibitem{PalaniswamiRaoBainbridge}
M.~Palaniswami, A.~S. Rao, and S.~Bainbridge, ``Real-time monitoring of the
  {Great Barrier Reef} using {Internet of Things} with big data analytics,''
  \emph{{ICT} Discoveries}, no. {Special Issue No. 1}, pp. 1--10, Oct. 2017.

\bibitem{YiLoMakLeung}
W.~Y. Yi, K.~M. Lo, T.~Mak, K.~S. Leung, Y.~Leung, and M.~L. Meng, ``A survey
  of wireless sensor network based air pollution monitoring systems,''
  \emph{Sensors}, vol.~15, pp. 31\,392--31\,427, 2015.

\bibitem{Zhao_industrial}
G.~Zhao, ``Wireless sensor networks for industrial process monitoring and
  control: A survey,'' \emph{Network Protocols and Algorithms}, vol.~3, no.~1,
  pp. 46--63, 2011.

\end{thebibliography}
\bibliographystyle{IEEEtran}

\end{document}